\documentclass{amsart}

\usepackage{amsmath}
\usepackage{amssymb}

\usepackage[utf8]{inputenc}
\usepackage[T1]{fontenc}

\def\dj{d\kern-.30em\raise1.25ex\vbox{\hrule width .3em height .03em}}
\def\Dj{D\kern-.70em\raise0.75ex\vbox{\hrule width .3em height .03em}
\kern.03em}

\makeatletter
\renewcommand{\subsection}{\@startsection{subsection}{2}{\z@}%
{\baselineskip}{0.5\baselineskip}{\bfseries}}
\def\l@section{\def\@tocpagenum##1{\hss{\bfseries ##1}}
\@tocline{1}{8pt}{0pc}{}{\bfseries}}
\def\l@subsection{\def\@tocpagenum##1{}
\@tocline{2}{2pt}{2pc}{2pc}{}}
\makeatother

\def\1{\varnothing}

\newtheorem{theorem}{Theorem}[section]
\newtheorem{lemma}{Lemma}[section]
\newtheorem{corollary}{Corollary}[section]
\newtheorem{prop}{Proposition}[section]
\theoremstyle{definition}
\newtheorem{definition}{Definition}[section]
\newtheorem{remark}[definition]{Remark}

\numberwithin{equation}{section}

\begin{document}

\renewcommand{\thepage}{\ifnum\value{page}=1 \else\arabic{page}\fi}

\title{Coherent States for the Manin Plane
\\	
via Toeplitz Quantization}

\author{Micho \Dj ur\dj evich}

\address{Universidad Nacional Aut\'onoma de M\'exico, 
Instituto de Matem\'aticas, Area de la Investigacion Cient\'{\i}fica, 
Circuito Exterior, Ciudad Universitaria, CP 04510,
Mexico City, Mexico}
\email{micho@matem.unam.mx}

\author{Stephen Bruce Sontz}

\address{Centro de Investigaci\'on en Matem\'aticas, A.C., 
(CIMAT), Jalisco S/N, Mineral de Valenciana, CP 36023, 
Guanajuato, Mexico}
\email{sontz@cimat.mx}

\begin{abstract}
	In the theory of Toeplitz quantization 
	of algebras, as developed by the second 
	author, coherent states are defined as
	eigenvectors of a 
	Toeplitz annihilation operator. 
	These coherent states are studied 
	in the case when the algebra is the 
	generically 
	non-commutative Manin plane. 
	In usual quantization schemes
	one starts with a classical phase space, 
	then quantizes it in order to produce  
	annihilation operators and then their eigenvectors 
	and eigenvalues. 
	But we do this in the opposite order, namely  
	the set of the 
	eigenvalues of the previously defined 
	annihilation operator 
	is identified as a 
	generalization of 
	a classical mechanical phase space.
    We introduce the resolution of the identity, 
	upper and lower symbols as well as 
	a coherent state quantization, which in turn 
	quantizes the Toeplitz 
	quantization. 
	We thereby have  
	a curious composition of
	quantization schemes. 
	We proceed by identifying a generalized
	Segal-Bargmann space $ \mathcal{SB} $
	of square-integrable, anti-holomorphic 
	functions as the image of a coherent state 
	transform. 
	Then $ \mathcal{SB} $ has a reproducing kernel function
	which allows us to define a secondary Toeplitz
	quantization, whose symbols are functions. 
    Finally, this is compared with
    the coherent states of the Toeplitz 
    quantization of a closely related 
    non-commutative space known as the 
    paragrassmann algebra.
\end{abstract}

\maketitle



\tableofcontents

\section{Introduction}

In \cite{sbs2} the second author has introduced a
general formalism for defining Toeplitz operators
whose symbols come from a not necessarily 
commutative algebra. 
One motivating example of this was presented 
in \cite{sbs1}. 
In the case when the algebra has a 
{\em $ * $-operation} 
(also known as a {\em conjugation})
we can often identify in a natural way a 
sub-algebra, which is {\em not} closed under 
the $ * $-operation, of 
{\em holomorphic} elements. 
Then its conjugate sub-algebra consists of the  
{\em anti-holomorphic} elements. 
The convention used here is that the common,
invariant domain of the 
Toeplitz operators is the 
holomorphic sub-algebra, which is also a 
pre-Hilbert space. 
The Toeplitz operators whose symbols are 
holomorphic elements then play the role 
of {\em creation operators} while the
Toeplitz operators whose symbols are 
anti-holomorphic elements play the role 
of {\em annihilation operators}. 

In some examples there is a natural grading 
of the elements in the sub-algebra
in which case the annihilation operators that
lower degree by $ 1 $ are analogous to the 
annihilation operator $ \partial / \partial z $ 
acting on holomorphic functions of the 
Segal-Bargmann space. 
(For details see \cite{barg}.) 
Such degree $ -1 $ annihilation operators can 
be used to define coherent states as their 
eigenvectors whose corresponding eigenvalues
are the points in a classical phase space. 
Turning this understanding around, we may use
an appropriate Toeplitz annihilation operator 
in order to define coherent states in the 
Toeplitz setting as its eigenvectors whose 
eigenvalues then define the classical phase space. 
We do exactly that in this paper 
in the setting given by the example in \cite{sbs1} 
of the non-commutative (or, as it is sometimes 
called, quantum) plane of Manin.
However, we wish to emphasize our viewpoint that 
the Manin plane is not a quantum object in the sense 
of quantum theory, since among other things 
Planck's constant $ \hbar > 0 $
does not enter into its structure. 
On the other hand, the Weyl-Heisenberg (unital)
algebra generated by elements $ P $ and $ Q $ and 
satisfying the commutation relation
$ i [P, Q ] = \hbar \, 1 $ does come from quantum 
theory, explicitly involves Planck's constant 
$ \hbar $ and 
has a natural action of 
the Lie group 
$ SL (2, \mathbb{C}) \rtimes \mathbb{C}^{2}$ (the semi-direct product)
that also acts naturally on the (complex) 
plane $ \mathbb{C}^{2}  $. 
So, the  non-commutative 
Weyl-Heisenberg algebra deserves to be
called the {\em quantum plane}. 
This point of view is also expressed in 
\cite{weaver}. 
Nonetheless, the Manin plane is an interesting 
object studied in 
non-commutative geometry. 

For background material on coherent states 
two excellent references 
with many examples are \cite{ali} and  
\cite{gazeau}. 
The recent book \cite{antoine} has 14 review  
papers on a variety of topics in the field. 
We use standard notation. 
For example, 
$ \mathbb{N} $,  $ \mathbb{Z} $, $ \mathbb{R} $, 
$ \mathbb{C} $ 
denote the sets of the non-negative 
integers,  all the  integers, the real numbers 
and  the complex numbers, respectively.
The complex conjugate of $ \lambda \in \mathbb{C} $
is denoted by $ \lambda^{*} $.

\section{The Setting}
\label{setting-section}

We consider the {\em Manin plane} 
$ \mathbb{C} Q_{q} (\theta, \overline{\theta})$
which is the unital, complex algebra generated by 
two conjugate elements $ \theta $ and $ \overline{\theta} $
subject only to the commutation relation 
$\theta \overline{\theta} = 
q \overline{\theta}      \theta $
for some non-zero $ q \in \mathbb{C} $. 
This algebra is  non-commutative 
except when $ q = 1 $. 
The elements 
$ \theta^{i} \overline{\theta}\,\!^{j} $
with $ i,j \in \mathbb{N} $ form a basis 
of 
$ \mathbb{C} Q_{q} (\theta, \overline{\theta}) $. 
Toeplitz operators with symbols in 
$ \mathbb{C} Q_{q} (\theta, \overline{\theta})$
have been defined and studied in \cite{sbs1}. 
In particular, for any element 
$g \in \mathbb{C} Q_{q} (\theta, \overline{\theta})$
there is a linear operator, called the 
{\em Toeplitz operator with symbol $ g $}, denoted by 
$ T_{g} : \mathcal{P} \to \mathcal{P} $, where
$ \mathcal{P} $ is the 
holomorphic subalgebra generated by 
the holomorphic variable $ \theta $. 
Of course, $ \mathcal{P} = \mathbb{C} [\theta]$ is 
the polynomial algebra generated by $ \theta $. 
So $ \mathcal{P} $ has a grading given by the degree 
of a homogeneous polynomial and 
has a basis $ \{ \theta^{n} \,|\, n \ge 0\} $. 
We now sketch how this {\em Toeplitz quantization} 
is realized. 
See \cite{sbs1} for more details.
 
There is a sesqui-linear form determined on 
$ \mathbb{C} Q_{q} (\theta, \overline{\theta})$ by 
\begin{equation}
\label{define-sesqui-form}
\langle 
\theta^{i} \overline{\theta}\,\!^{j}, 
\theta^{k} \overline{\theta}\,\!^{l}
 \rangle := w_{i+l}\, \delta_{i-j,k-l}
 \quad \mathrm{for~} i,j,k,l \in \mathbb{N}, 
 \end{equation}
where the {\em weights} $ w_{n} $ satisfy 
$ w_{n} > 0 $ for every integer $ n \ge 0 $. 
Also, $ \delta_{n,m} $ denotes the Kronecker delta
for integers $ n,m \in \mathbb{Z} $. 
Our convention is that a sesqui-linear form
is linear is its second entry while being anti-linear 
in its first entry. 
The same convention holds for the inner product
in a Hilbert space. 

When restricted to $ \mathcal{P} $ this 
sesqui-linear form \eqref{define-sesqui-form}
satisfies  
$$
\langle \theta^{i} , \theta^{k} \rangle = 
w_{i} \, \delta_{i,k},
$$
which clearly gives a positive definite inner product. 
So, 
$$ 
\mathcal{B}:= 
\Big\{ \phi_{j} \equiv \phi_{j} (\theta) := 
\frac{1}{w_{j}^{1/2}} \theta^{j} ~\Big|~ j \ge 0  \Big\} 
$$ 
is an orthonormal Hamel basis for the 
(incomplete) pre-Hilbert space 
$ \mathcal{P} $. 
We let $ \mathcal{H} $ denote the Hilbert space 
completion of $ \mathcal{P} $. 
The inner product on $ \mathcal{H} $, being the 
extension of that on $ \mathcal{P} $, is also 
denoted as $ \langle \cdot , \cdot \rangle $. 
We note that $ \mathcal{B} $ is an orthonormal basis
for $ \mathcal{H} $. 
We use the sesqui-linear form to define a 
linear map $ P $ on 
$ \mathbb{C} Q_{q} (\theta, \overline{\theta}) $ 
determined by 
$$
P (\theta^{i} \overline{\theta}\,\!^{j} ) :=
\sum_{k=0}^{\infty} w_{k}^{-1} \langle
\theta^{k}, \theta^{i} \overline{\theta}\,\!^{j} \rangle \, \theta^{k} 
\quad \mathrm{for~all~}i,j \in \mathbb{N}. 
$$
Due to the definition of the sesqui-linear form, 
this infinite sum has at most one 
non-zero term (when $ k = i - j \ge 0 $)
and so it makes sense. 
Since the elements 
$ \theta^{i} \overline{\theta}\,\!^{j} $
with $ i,j \in \mathbb{N} $ form a basis 
of 
$ \mathbb{C} Q_{q} (\theta, \overline{\theta}) $, 
this uniquely determines the linear map $ P $. 
In Dirac notation (which technically does not apply,
since we are not in a Hilbert space setting) we 
have that 
$ P = \sum_{k=0}^{\infty} | \phi_{k} (\theta) \rangle
\langle \phi_{k} (\theta) |  $. 
So it is not surprising that $ P = P^{2} $ (that is,
$ P $ is a projection) with range $ \mathcal{P} $, 
the algebraic span of the $ \phi_{k} (\theta)$. 

Now for any 
$g \in \mathbb{C} Q_{q} (\theta, \overline{\theta})$
we define the 
{\em Toeplitz operator with symbol $ g $}, denoted as
$ T_{g} : \mathcal{P} \to \mathcal{P} $, by
$ T_{g} (\phi) := P (g \phi) $ for all 
$ \phi \in \mathcal{P} $.  
Notice that the product $g \phi $ of the 
two elements $g, \phi \in 
\mathbb{C} Q_{q} (\theta, \overline{\theta}) $ 
is again an element in the algebra 
$\mathbb{C} Q_{q} (\theta, \overline{\theta})$. 
Then the projection $ P $ maps this product to 
an element of $ \mathcal{P} $. 
We have chosen to multiply the symbol $ g $ 
on the left in the definition of $ T_{g} $. 
A similar, but not identical,  theory 
entails if we multiply 
on the right, which was what was done in \cite{sbs1}. 
We will explain later on why we 
preferred using multiplication on the left in 
this paper. 

As explained in detail in \cite{sbs2} this 
theory of densely defined Toeplitz operators,  
acting in a Hilbert space,
gives a 
quantization scheme, called {\em Toeplitz
quantization}, that includes Planck's 
constant as well as creation and 
annihilation operators. 
So this is a quantum theory. 

The operators 
$ T_{\theta^{i} \overline{\theta}\,\!^{j}} $ 
were explicitly calculated in Theorem~4.4 in 
\cite{sbs1}, where we had multiplication of the 
symbol on the right in the definition of 
Toeplitz operators. 
Since we have used here multiplication of the symbol 
on the left, we have a different result. 
\begin{theorem}
\label{T-action}
For integers $ i,j,n \ge 0 $ we have 
\begin{equation}
\label{T-theta-i-theta-bar-j}
T_{\theta^{i} \overline{\theta}\,\!^{j}} 
\phi_{n}(\theta) = 
q^{-j n} 
\dfrac{w_{n+i}}{( w_{n} w_{n+i-j} )^{1/2}} 
\phi_{n+i-j} (\theta),
\end{equation}
where for $ m < 0 $ we put 
$ \phi_{m} (\theta) = 0 $ and $ w_{m} = 1 $. 
\end{theorem}
\begin{proof}
We use the reproducing kernel {\em object} 
$ K (\theta, \eta) = 
\sum_{k} \phi_{k} 
(\overline{\theta}) \otimes \phi_{k} (\eta) $, 
which is a tensor but not a {\em function}, to 
calculate the projection operator. 
Here $ \eta, \overline{\eta} $ is another, 
independent pair of variables satisfying 
$ \eta \overline{\eta} = q \overline{\eta}\eta$. 
Also, $ \langle \cdot , \cdot \rangle_{\eta} $ 
denotes the same sesqui-linear form as above, but with
respect to the new variables 
$ \eta, \overline{\eta} $. 
See \cite{sbs1} for more details. 
We calculate 
\begin{align*}
T_{\theta^{i} \overline{\theta}\,\!^{j}} 
 \phi_{n}(\theta)
&=  
\big\langle 
K (\theta, \eta) , 
\eta^{i} \overline{\eta}\,\!^{j} \phi_{n} (\eta)
\big\rangle_{\eta} 
\\
&= 
\big\langle 
\sum_{k} \phi_{k} (\overline{\theta}) \otimes 
\phi_{k} (\eta), 
w_{n}^{-1/2}  \eta^{i} \overline{\eta}\,\!^{j} \eta^{n}
\big\rangle_{\eta} 
\\
&= 
w_{n}^{-1/2} 
\big\langle 
\sum_{k} \phi_{k} (\overline{\theta}) \otimes 
\phi_{k} (\eta), 
q^{-j n} \eta^{i+n} \overline{\eta}\,\!^{j} 
\big\rangle_{\eta} 
\\
&= 
q^{-j n} w_{n}^{-1/2} 
\big\langle 
\sum_{k} 
w_{k}^{-1/2} \eta^{k}, \, 
\eta^{i+n} \overline{\eta}\,\!^{j} 
\big\rangle \phi_{k} (\theta)
\\
&= 
q^{-j n} w_{n}^{-1/2} 
\sum_{k} w_{k}^{-1/2}  \delta_{k+j, i+n} w_{k+j}
\phi_{k} (\theta)
\\
&= 
q^{-j n} w_{n}^{-1/2} w_{n+i-j}^{-1/2} w_{i+n} 
\phi_{n+i-j} (\theta). 
\end{align*}
Even though the sum is over all $ k \in \mathbb{N} $, 
only at most one term is non-zero, namely 
when $ k = n + i - j \ge 0 $. 
\end{proof}
\begin{remark}
So $ T_{\theta^{i} \overline{\theta}\,\!^{j}} $ 
changes the degree of the monomial $ \phi_{n} $
by $ i-j $ or annihilates it.  
Note that this result differs from Theorem~4.4 
of \cite{sbs1} only by the factor of $ q^{-j n} $. 
There are 
two special cases of interest. 
One is the 
{\em annihilation operator} $ T_{\overline{\theta}} $
with anti-holomorphic symbol $ \overline{\theta} $ 
when $ i=0 $ and $ j=1 $. 
And the other case is the {\em creation operator} 
$ T_{\theta} $ with holomorphic symbol $ \theta $ 
when $ i = 1 $ and $ j = 0 $. 
\end{remark}

We turn our attention to the annihilation 
operator about which \eqref{T-theta-i-theta-bar-j} 
says that 
\begin{equation}
\label{T-theta-bar}
   T_{\overline{\theta}} \, \phi_{n}= 
   q^{-n}
   \left( \dfrac{w_{n}}{w_{n-1}} \right)^{1/2}
   \!\!\!\phi_{n-1}
   \quad \mathrm{for~every~integer~} n \ge 1 
\end{equation}
and that 
$ T_{\overline{\theta}} \, \phi_{0} =0 $. 
This defines $ T_{\overline{\theta}} $ 
on the subspace $ \mathcal{P} $, which is
dense in the Hilbert space $ \mathcal{H} $ and 
is invariant under the action of 
$ T_{\overline{\theta}} $. 
Then $ T_{\overline{\theta}} $ can be extended 
in the following way to the dense domain 
$ \mathrm{D} (T_{\overline{\theta}})  $
of elements 
$ \phi = \sum_{n=0}^{\infty} a_{n} \phi_{n} $
such that  
$ \sum_{n=0}^{\infty} a_{n} T_{\overline{\theta}} 
\, \phi_{n} $ 
is convergent in $ \mathcal{H} $, that is to say,
\begin{equation}
\label{define-dom-annihilation}
     \mathrm{D} (T_{\overline{\theta}}):= 
     \Big\{ \phi = \sum_{n=0}^{\infty} 
     a_{n} \phi_{n} \in \mathcal{H} 
     ~\Big|~ 
     \sum_{n=1}^{\infty} 
     |a_{n}|^{2} \, |q|^{-2n}
     \left( 
        \dfrac{w_{n}}{w_{n-1}} 
     \right) 
     < \infty \Big\} 
\end{equation}
and 
$$
T_{\overline{\theta}} \, \phi := 
\sum_{n=1}^{\infty} a_{n} \, q^{-n}
\left( \dfrac{w_{n}}{w_{n-1}} \right)^{1/2}
\phi_{n-1}
$$
for every $ \phi = \sum_{n=0}^{\infty} a_{n} \phi_{n} 
\in \mathrm{D} (T_{\overline{\theta}}) $. 
Actually, by standard functional analysis 
this extension is the closure of 
$ T_{\overline{\theta}} $ defined on the domain 
$ \mathcal{P} $. 

\vskip 0.1cm
Then we have 
$ \mathcal{P} \subsetneq 
\mathrm{D} (T_{\overline{\theta}}) 
\subset \mathcal{H}$. 
The last inclusion is an equality if and only if 
the sequence 
$\{ |q|^{-2 n} \, (w_{n}/w_{n-1}) ~|~ n \ge 1 \}$ 
is bounded if and only if  $ T_{\overline{\theta}} $ 
is a bounded operator. 
Also,  $ T_{\overline{\theta}} $ is compact if and 
only if 
$ \lim_{n \to \infty} |q|^{-2 n} \, (w_{n}/w_{n-1})=0 $. 

The choice $ w_{n} = n! $ for the weights 
is motivated by the
Segal-Bargmann space. (See \cite{barg}.) 
With this choice 
$  T_{\overline{\theta}} \, \phi_{n} = 
q^{-n} \, n^{1/2} \, \phi_{n-1} $ for $ n \ge 1 $. 
If also $ q = 1 $, then 
$  T_{\overline{\theta}} $ is an unbounded operator, 
which is unitarily equivalent to the annihilation 
operator 
$ \partial / \partial z $
in the Segal-Bargmann space.
Another choice is $ w_{n} = 1 $ in which case 
$  T_{\overline{\theta}} $ is a  
{\em weighted backwards shift operator}, a bounded 
operator if and only if $ |q| \ge 1 $.

\section{Coherent States}

We now arrive at our basic definition.

\begin{definition}
A {\em coherent state} for this Toeplitz quantization
of the Manin plane is an eigenvector
of the annihilation 
operator $ T_{\overline{\theta}} $. 
More explicitly, it is a non-zero vector 
$ \phi_{\lambda} \in 
\mathrm{D} (T_{\overline{\theta}}) $ such 
that
$ T_{\overline{\theta}} \, \phi_{\lambda} =
\lambda \phi_{\lambda} $, 
where 
$ \lambda \in \mathbb{C}$ is 
the eigenvalue. 
The set of eigenvalues of $ T_{\overline{\theta}} $ 
is called the {\em phase space}. 
\end{definition}

\begin{remark}
	The terminology `phase space' comes from
	the theory of  
	classical mechanics. 
	But 
	in functional analysis the set of eigenvalues 
	of a densely defined linear operator 
	is called its {\em point spectrum,} at least by 
	some authors.  
	An important feature of this approach 
	is that the (possibly empty!) 
	phase space is characterized purely 
	by properties of the quantum theory, that is to
	say, the phase space is a quantum object. 
    N.B.:  
	We do not start with a phase space and then 
	quantize it. 
\end{remark}

We now find the coherent states by 
taking $ \phi = \sum_{n=0}^{\infty} a_{n} \phi_{n}$,
where the unknown coefficients must satisfy 
$  \sum_{n=0}^{\infty} |a_{n}|^{2} < \infty $,  
and noting that 
$ T_{\overline{\theta}} \, \phi = \lambda \phi $, 
where $ \lambda \in \mathbb{C} $, becomes
$$
T_{\overline{\theta}} \, \phi = 
\sum_{j=1}^{\infty} a_{j} \, q^{-j}
\left( \dfrac{w_{j}}{w_{j-1}} \right)^{1/2}
\!\!\!\!\phi_{j-1} =
\sum_{n=0}^{\infty} \lambda a_{n} \phi_{n}.
$$
Putting $ n = j - 1 $ in the first sum gives 
$$
\sum_{n=0}^{\infty} a_{n+1} \, q^{-(n+1)}
\left( \dfrac{w_{n+1}}{w_{n}} \right)^{1/2}
\!\!\!\!\phi_{n} =
\sum_{n=0}^{\infty} \lambda a_{n} \phi_{n}.
$$
Using orthogonality we see that a necessary and
sufficient condition for this equality 
to hold is that for all integers $ n \ge 0 $ we have 
\begin{equation}
\label{recursion}
a_{n+1} = \lambda \, q^{n+1}
\left( \dfrac{w_{n}}{w_{n+1}} \right)^{1/2} 
\!\!\!\! a_{n}. 
\end{equation}
Therefore, up to a non-zero 
multiplicative constant, there is at
most one coherent state with eigenvalue 
$ \lambda \in \mathbb{C} $. 
In fact, the recursion relation 
\eqref{recursion} is solved explicitly by 
$$
       a_{n} = \lambda^{n} \, q^{n(n+1)/2}
       \left( \dfrac{w_{0}}{w_{n}} \right)^{1/2} 
       \!\!\!\! a_{0}
       \quad \mathrm{for~all~integers~} n \ge 0, 
$$
where $ a_{0} \in \mathbb{C} \setminus \{ 0 \} $ is arbitrary. 
(We exclude $ a_{0} =0 $ since an eigenvector 
by definition must be non-zero.)
Moreover, in the case when $ \lambda \ne 0 $ 
we see that $ a_{n} \ne 0 $ for all $ n \ge 1 $
and therefore 
$ \phi = \sum_{0}^{\infty} a_{n} \phi_{n} $ 
can not be an
element of $ \mathcal{P} $. 
This is why we introduced the 
larger domain $ \mathrm{D} (T_{\overline{\theta}}) $. 
We will return to this point in a moment. 

For convenience we simplify by taking 
$ a_{0} = w_{0}^{-1/2} $ thereby getting 
$$
a_{n} = \lambda^{n} \, q^{n(n+1)/2} \,  w_{n}^{-1/2} 
\quad \mathrm{for~all~integers~} n \ge 0.
$$
And so, up to a multiplicative non-zero constant, 
the coherent state with eigenvalue 
$ \lambda \in \mathbb{C} $ has to be 
\begin{equation}
\label{coherent-state}
\phi_{\lambda} := \sum_{n=0}^{\infty} 
                 \lambda^{n} \, q^{n(n+1)/2}
                 \, w_{n}^{-1/2} \phi_{n}
\end{equation}
{\em provided at least} that the series 
converges in $ \mathcal{H} $. 
And if this series diverges, then there is no
coherent state with eigenvalue equal to $ \lambda $. 
Of course, this series is convergent 
by Hilbert space theory 
if and only if 
\begin{equation}
\label{CS-norm-squared}
|| \phi_{\lambda} ||^{2} = 
\sum_{n=0}^{\infty} |\lambda|^{2n} \, |q|^{n(n+1)}
\, w_{n}^{-1}
< \infty, 
\end{equation}
which is a power series in the variable 
$ |\lambda|^{2} $. 
We let $ w $ denote  
$\{ w_{n} \,|\, n \ge 0 \}$, the sequence of the weights. 
Then, using the theory of power series,  
the radius of convergence $ R_{w} $ of this 
{\em as a power 
series in $ | \lambda | $} is given by the 
formula 
$$
   R_{w} = \Big( 
 \limsup_{n \to \infty}  
  |q|^{(n+1)} w_{n}^{-1/n} \Big)^{-1/2} 
  \!\!= \,\, \Big( 
  \liminf_{n \to \infty} 
  |q|^{-(n+1)} w_{n}^{1/n} \Big)^{1/2}
  \!\!\!\!\!\!, 
$$
where we use the standard conventions 
$ 0^{-1} = \infty $ and $ \infty^{-1} = 0 $. 
So, the formula $ \eqref{coherent-state} $ defines
a vector $ \phi_{\lambda} \in \mathcal{H} $ for
all $ \lambda \in \mathbb{C} $ 
satisfying $ |\lambda| < R_{w} $. 
Also, the series in \eqref{coherent-state} 
diverges if $ |\lambda| > R_{w} $, in which case 
there is no coherent state 
with eigenvalue~$ \lambda $.  
It is worthwhile to note that the infinite series in 
\eqref{CS-norm-squared} converges for some 
complex number $ \lambda $ if and only if 
that series converges at every point on the circle 
$ \lambda \, e^{i \theta} $ of radius $ |\lambda| $, 
where $ \theta \in \mathbb{R} $. 
The associated holomorphic function 
$$
   f(z):= \sum_{n=0}^{\infty} z^{n} \, |q|^{n(n+1)}
   \, w_{n}^{-1},
$$
where $ z $ is a complex variable,  
has radius of convergence $ R_{w}^{2} $.
Suppose that $ R_{w} < \infty $. 
Since this series has 
positive coefficients, it converges absolutely 
at every point on the circle 
$ |z| = R_{w}^{2} $ 
if and only if it converges absolutely 
at one point on that circle if and only if 
it converges (absolutely of necessity) 
at $ z = R_{w}^{2} $.  

It is not enough that $ \phi_{\lambda} $ 
is an element of the Hilbert space $ \mathcal{H} $, 
even though that is a necessary condition.  
It must also belong to 
$\mathrm{D} (T_{\overline{\theta}})$, 
the domain of the annihilation operator,  
as noted earlier, 
and satisfy $ T_{\overline{\theta}} \, \phi_{\lambda} = 
\lambda \, \phi_{\lambda} $. 
We now explicitly prove this. 
\begin{prop}
	For all $ \lambda $ satisfying 
	\eqref{CS-norm-squared} (in particular, if
	$ |\lambda| < R_{w} $) we have that  
	$ \phi_{\lambda} \in 
	\mathrm{D} (T_{\overline{\theta}}) $. 
	Moreover, 
	$ T_{\overline{\theta}} \, \phi_{\lambda} = 
	\lambda \, \phi_{\lambda} $ and 
	$ \phi_{\lambda} \ne 0 $
	for all such $ \lambda \in \mathbb{C} $.
\end{prop} 
\begin{proof}
According to \eqref{define-dom-annihilation} and 
\eqref{coherent-state}, proving  
$ \phi_{\lambda} \in 
\mathrm{D} (T_{\overline{\theta}}) $
is equivalent to showing that
the following expression is finite:
\begin{align*}
	\sum_{n=1}^{\infty} 
	|q|^{-2 n}
	&\left( \dfrac{w_{n}}{w_{n-1}} \right) 
	\big| 
	\lambda^{n} q^{n(n+1)/2} w_{n}^{-1/2} 
	\big|^{2} 
	= \sum_{n=1}^{\infty} 
	\dfrac{|q|^{-2n + n(n+1)}}{w_{n-1}} |\lambda|^{2n}
	\\
	&= |\lambda|^{2}
	\sum_{n=1}^{\infty} 
	\dfrac{|q|^{n(n-1)}}{w_{n-1}}
    |\lambda|^{2(n-1)} 
    = |\lambda|^{2} \, 
    || \phi_{\lambda} ||^{2}
	< \infty, 
\end{align*}
where the inequality holds 
because $ \lambda $ satisfies 
\eqref{CS-norm-squared}. 
Since 
$ \langle \phi_{0}, \phi_{\lambda} \rangle = w_{0}^{-1/2} \ne 0 $ by 
\eqref{coherent-state}, we see that 
$ \phi_{\lambda} \ne 0 $. 
Finally, we prove $ \phi_{\lambda} $ is an 
eigenvector whose eigenvalue is $ \lambda $  
by calculating 
\begin{align*}
T_{\overline{\theta}} \, \phi_{\lambda} &= 
T_{\overline{\theta}} \, \left( 
\sum_{n=0}^{\infty} \lambda^{n} \, q^{n(n+1)/2}
 \, w_{n}^{-1/2} \, 
\phi_{n} \right) 
\\
&= 
\sum_{n=1}^{\infty} \lambda^{n} \, q^{n(n+1)/2}
\, w_{n}^{-1/2} q^{-n}
\left( \dfrac{w_{n}}{w_{n-1}} \right)^{1/2} \!\!\!\!
\phi_{n-1} 
\\
&= \lambda \sum_{n=1}^{\infty} 
\lambda^{n-1} \, q^{n(n-1)/2}
w_{n-1}^{-1/2} \, \phi_{n-1}
\\
&=
\lambda \, \phi_{\lambda}, 
\end{align*}
which finishes what we wanted to prove. 
\end{proof}

For $ R_{w} > 0 $ 
the phase space consists of the open ball 
$ |\lambda| < R_{w}$ 
plus possibly the points on its boundary 
for the case $ R_{w} < \infty $. 
The case $ R_{w} = 0 $ will be discussed in a moment. 
Clearly, \eqref{CS-norm-squared} converges 
at some point on the boundary if and only if 
it converges at all points on the boundary. 
We let $ B_{w} $ denote  
the phase space; this is either 
the open or closed ball of radius $ R_{w} $. 
It seems reasonable to conjecture that the 
spectrum of $ T_{\overline{\theta}} $ is the closure
of $ B_{w} $, that is,  
$ \mathrm{Spec} (T_{\overline{\theta}}) = 
\overline{B_{w}}$. 
We note the possibility that some elements on 
the boundary of $ B_{w} $ may not be 
eigenvalues.  

The case $ w_{n} = n! $ and $ |q| = 1 $
arises as noted above in analogy with 
the situation in the Segal-Bargmann space. 
Using the Stirling approximation for $ n! $
we see 
for this case that
$$
   R_{w} = 
   \lim_{n \to \infty}  
   \left( 
   (2 \pi)^{1/2} n^{n + 1/2} e^{-n}
   \right)^{1/2 n} 
   = 
   \lim_{n \to \infty} 
    (2 \pi)^{1/4n} n^{1/2 + 1/4n} e^{-1/2} 
    = \infty.
$$
Thus the phase space is $ \mathbb{C} $, and the
annihilation operator has point 
spectrum equal to the entire complex plane. 

The case $ w_{n} = 1 $ for all integers $ n \ge 0 $
and $ q = 1 $
clearly leads to $ R_{w} = 1$, which is the
spectral radius of the backwards shift operator.  
By picking other values for the weights
and for $ q $ we can find 
any value of $ R_{w} $ in $ [0, \infty] $. 

In this paper the Manin plane 
$ \mathbb{C} Q_{q} (\theta, \overline{\theta})$ 
plays the role of the `classical phase 
space' 
that is being quantized, but the 
phase space $ B_{w} $ determined by 
the quantum theory 
is quite unlike 
$ \mathbb{C} Q_{q} (\theta, \overline{\theta})$ 
unless $ R_{w} = \infty $ and $ q = 1 $. 
The fact that the Manin plane 
$ \mathbb{C} Q_{q} (\theta, \overline{\theta})$
has the ball $ B_{w} $ as its phase space is quite 
surprising. 
Of course, \eqref{coherent-state} always 
converges for $ \lambda =0 $. 
But in the case when $ R_{w} = 0 $ that is the
only complex number $ \lambda $ for which 
\eqref{coherent-state} converges and so the
phase space has exactly one point, a rather
curious quantum situation with a
trivial phase space. 
This motivates the following idea. 
\begin{definition}
\label{define-extreme-quantum}
We say that a quantum 
theory which has exactly 
one annihilation operator $ A $ 
is an {\em extreme quantum theory} if the point
spectrum of $ A $ has at most one point, that is 
to say, the corresponding phase space 
has at most one point. 
\end{definition}

\begin{remark}
Since one typically has $ \ker \, A \ne 0 $ 
for an annihilation operator $ A $, the complex
number $ 0 $ is then in the point spectrum 
of $ A $. 
While it is mathematically possible to 
have $ \ker \, A = 0 $, this may be
undesirable from a physics viewpoint. 
Until Section~\ref{comparison-section} of this paper 
we will assume that $ R_{w} > 0 $. 
So the interior of the phase space $ B_{w} $
is  non-empty. 
However, the case $ R_{w} = 0 $ is fascinating, 
though we currently have few mathematical tools
for studying it. 
For this reason, and this reason alone, 
we exclude it from consideration here. 
\end{remark}

\section{Resolution of the Identity}
\label{section-resolution-identity}

In order to study the phase space $ B_{w} $ we note 
that the coherent states define a parametrized 
family of rank-one projection operators  
$ | \phi_{\lambda} \rangle \langle \phi_{\lambda} | $
in Dirac notation, where $ \lambda \in B_{w} $. 
We might want to find a positive Borel measure 
$ \rho $ on $ B_{w} $ such that
\begin{equation}
\label{resolution-of-id}
I = \int_{B_{w}} d \rho(\lambda) \,  
 | \phi_{\lambda} \rangle \langle \phi_{\lambda} |
\end{equation}
where $ I $ is the identity operator acting 
in $ \mathcal{H} $. 
This is a giant step beyond 
Toeplitz quantization whose virtue 
is that it does not use a measure. 
Equation \eqref{resolution-of-id} is called 
the {\em resolution of the identity}, 
and the integral in it is to be understood 
with respect to the weak operator topology, which means
that for all $ \phi, \psi \in \mathcal{H} $ we have
\begin{equation}
\label{i-p-resolution}
 \langle \phi, \psi \rangle 
= \int_{B_{w}} d \rho(\lambda) \, 
  \langle \phi, 
 \phi_{\lambda} \rangle 
 \langle \phi_{\lambda} , \psi \rangle. 
 \end{equation}
The integrand is a measurable function 
of $ \lambda \in B_{w} $. 
(See Remark~\ref{remark-section-8}
for why this is true.)
We are requiring here that it is also  
absolutely integrable 
with respect to the unknown measure $ \rho $ 
for all $ \phi, \psi \in \mathcal{H} $. 
While this is the standard definition 
of a resolution of the identity, 
the reader should be aware that this is an 
extremely strong condition on $ \rho $. 
In particular, 
this imposes a lot of necessary conditions on 
the measure $ \rho $. 
For example, by 
putting $ \phi = \psi = \phi_{n} $, 
a standard basis element, we see that
\begin{equation}
\label{necessary-condition-on-rho}
\int_{B_{w}} d \rho(\lambda) 
\, |\lambda|^{2n} |q|^{n(n+1)} w_{n}^{-1} = 1
\quad \mathrm{for~every~integer~}  n \ge 0 . 
\end{equation}
The case $ n = 0 $ says that $ w_{0}^{-1} \rho $ 
is a probability measure, that is,
$E \mapsto  \int_{E} d \rho(\lambda) \, w_{0}^{-1}  $ 
for Borel subsets $ E \subset B_{w} $
is a probability measure on the phase space $ B_{w} $. 
Also \eqref{necessary-condition-on-rho} gives us
necessary conditions on the even complex 
moments of that probability measure, namely
for every integer $ n \ge 0 $ we must have 
$$
\big|
\int_{B_{w}} d \rho(\lambda) w_{0}^{-1} \lambda^{2n}
\big| 
 \le  
\int_{B_{w}} d \rho(\lambda) w_{0}^{-1} |\lambda|^{2n} 
= |q|^{-n(n+1)} w_{n} w_{0}^{-1}. 
$$
It may well be that no such measure $ \rho $ exists, 
but even so the Toeplitz quantization has 
its own intrinsic interest.
The failure of the necessary conditions would imply
that is the case. 
Furthermore, even if such a measure exists, it may 
not be unique. 

We expand any pair 
$ \phi, \psi \in \mathcal{H} $
in the orthonormal basis to get
 $ \phi = \sum_{j} a_{j} \phi_{j}$
and $ \psi = \sum_{k} b_{k} \phi_{k}$ 
with complex numbers $ a_{j}, b_{k} $ 
satisfying 
$ || \psi ||^{2} = \sum_{j} |a_{j}|^{2} < \infty $ 
and $ || \phi ||^{2} = \sum_{k} |b_{k}|^{2} < \infty $. 
So we obtain  
\begin{align}
\nonumber 
\int_{B_{w}} \!\! d \rho(\lambda) \, 
\langle \phi, 
\phi_{\lambda} \rangle 
\langle \phi_{\lambda} , \psi &\rangle 
= 
\int_{B_{w}} \!\! d \rho(\lambda) \, 
\sum_{j,k=0}^{\infty} 
a_{j}^{*} \lambda^{j} \, q^{j(j+1)/2}
w_{j}^{-1/2}
b_{k} \lambda^{*k} \, q^{*k(k+1)/2}
w_{k}^{-1/2}
\nonumber
\\
\label{expanding-out}
&= 
\sum_{j,k=0}^{\infty} 
a_{j}^{*} q^{j(j+1)/2} w_{j}^{-1/2}
b_{k} q^{*k(k+1)/2} w_{k}^{-1/2} 
\int_{B_{w}} d \rho(\lambda) \, 
\lambda^{j} \lambda^{*k}. 
\end{align}
The interchange of infinite sum and integral 
in the last equality 
needs to be justified. 
This is a delicate situation, which we now 
examine. 
We first estimate the partial sums of the infinite 
sum under the integral sign. 
We take integers $ N,M \ge 0 $ getting 
\begin{align*}
\Big| 
&\sum_{j,k=0}^{N,M} 
a_{j}^{*} \lambda^{j} q^{j(j+1)/2}
w_{j}^{-1/2}
b_{k} \lambda^{k*} q^{*k(k+1)/2}
w_{k}^{-1/2} 
\Big| 
\\
&\le 
\sum_{j,k=0}^{N,M} 
|a_{j}| |\lambda|^{j} |q|^{j(j+1)/2}
w_{j}^{-1/2}
|b_{k}| |\lambda|^{k} |q|^{k(k+1)/2}
w_{k}^{-1/2}
\\
&\le 
\sum_{j=0}^{\infty} 
|a_{j}| |\lambda|^{j} |q|^{j(j+1)/2}
w_{j}^{-1/2}
\sum_{k=0}^{\infty} 
|b_{k}| |\lambda|^{k} |q|^{k(k+1)/2}
w_{k}^{-1/2}
\\
&\le 
\left( 
\sum_{j=0}^{\infty} 
|a_{j}|^{2} 
\right)^{1/2}
\left( 
\sum_{k=0}^{\infty} 
|b_{k}|^{2}
\right)^{1/2}
\left(  
\sum_{l=0}^{\infty}
|\lambda|^{2l} |q|^{l(l+1)} w_{l}^{-1}
\right)^{2/2} \!\!\! 
= 
|| \phi || \, || \psi || \, || \phi_{\lambda} ||^{2}.
\end{align*}
We got the last estimate by applying 
Cauchy-Schwarz twice, which accounts for the 
exponent $ 2(1/2) = 2/2 $ in the last factor. 
To apply the 
Lebesgue  dominated convergence theorem 
now we would need to prove that 
$ || \phi_{\lambda} ||^{2} $ as a function 
of $ \lambda \in B_{w} $ is integrable 
with respect to the measure $ \rho $. 
But that is false! 
In fact,
\begin{align}
\label{divergent-integral}
\int_{B_{w}} d \rho(\lambda) \, 
|| \phi_{\lambda} ||^{2} &= 
\int_{B_{w}} d \rho(\lambda) \, 
\sum_{n=0}^{\infty} |\lambda|^{2n} \, |q|^{n(n+1)}
\, w_{n}^{-1} 
\\
&= \sum_{n=0}^{\infty} \int_{B_{w}} d \rho(\lambda) \, 
|\lambda|^{2n} \, |q|^{n(n+1)}
\, w_{n}^{-1} 
= \sum_{n=0}^{\infty} 1 = \infty, 
\nonumber 
\end{align}
where we used \eqref{CS-norm-squared}, 
the monotone convergence theorem 
and \eqref{necessary-condition-on-rho} in 
that order. 
Therefore the 
interchange does not follow so easily from the 
Lebesgue  dominated convergence theorem. 
One can change this argument 
in standard ways 
so that it only applies 
to $ \phi, \psi $ in some dense subspace, and 
thereby justify the interchange for 
such $ \phi, \psi $.
The easiest change by far of this type 
is to suppose that 
$ \phi, \psi \in \mathcal{P}$ so that the infinite sum 
collapses to a finite sum and therefore
\eqref{expanding-out} is trivially true. 

We next assume that the measure $ \rho $ is radial 
and absolutely continuous with respect to 
Lesbegue measure on $ B_{w} $. 
So it is equal to $ \rho (r) \, r d r d \alpha $ 
in the standard polar coordinates $ r, \alpha $. 
By a common abuse of notation 
we also let $ \rho $ denote here a measurable function 
$ \rho : [0, R_{w}] \to [0, \infty) $. 
The meaning of $ \rho $ will be clear 
from context. 
Radial symmetry is 
an enormously severe restriction  which 
could well eliminate from consideration 
many interesting cases which deserve further 
study. 
In particular if $ B_{w} $ contains 
its boundary $ \partial B_{w} $, this restriction 
together with  
the absolute continuity implies 
that $ \rho(\partial B_{w}) = 0 $.  
We make it in order to facilitate calculations 
such as the following which continues from 
\eqref{expanding-out}.

We assume that 
$ \phi, \psi \in \mathcal{P}$ and calculate that 
\begin{align*}
\int_{B_{w}} \!\! &d \rho(\lambda) \, 
\langle \phi, 
\phi_{\lambda} \rangle 
\langle \phi_{\lambda} , \psi \rangle = 
\sum_{j,k=0}^{\infty} \! 
a_{j}^{*} q^{j(j+1)/2} w_{j}^{-1/2}
b_{k} q^{*k(k+1)/2} w_{k}^{-1/2} 
\int_{B_{w}} \!\! d \rho(\lambda) \, 
\lambda^{j} \lambda^{k*}
\\
&=
\sum_{j,k=0}^{\infty} 
a_{j}^{*} q^{j(j+1)/2} w_{j}^{-1/2}
b_{k} q^{*k(k+1)/2} w_{k}^{-1/2} 
\int_{0}^{2\pi} d \alpha \int_{0}^{R_{w}} 
r \, dr \rho(r) \, r^{j+k} e^{i (j-k) \alpha}
\\
&= 
\sum_{j=0}^{\infty} 
a_{j}^{*}  b_{j} |q|^{j(j+1)} w_{j}^{-1} \, 
2 \pi \int_{0}^{R_{w}} 
dr \rho(r) \, r^{2j+1}. 
\end{align*}
By the resolution of the identity  
\eqref{i-p-resolution} this has to be equal to 
$ \langle \phi, \psi \rangle = 
\sum_{j=0}^{\infty}  a_{j}^{*}  b_{j}$ 
for all $ \phi, \psi \in \mathcal{P} $, that is,
for all sequences $ \{ a_{j} \,|\, j \ge 0 \} $ and 
$ \{ b_{j} \,|\, j \ge 0 \}$ with only finitely 
many non-zero terms. 
Therefore, by taking 
$ \phi = \psi = \phi_{n} \in \mathcal{P} $, 
we have the following result. 

\begin{theorem}
Suppose that the resolution of the identity 
\eqref{resolution-of-id} holds for a
measure $ \rho $ 
on $ B_{w} $ which is radial 
and absolutely continuous with respect to 
Lebesgue measure.  
Then necessary 
conditions on the `radial' function 
$ \rho : [0, R_{w}] \to [0, \infty) $ are: 
\begin{equation}
\label{first-rho-condition}
\int_{0}^{R_{w}} d r \, \rho(r) \, r^{2j+1} 
=  |q|^{-j(j+1)} \, w_{j}/2 \pi 
\quad \mathrm{for~all~} j \ge 0.  
\end{equation}
\end{theorem}

\begin{remark}
We change variables in the integral 
\eqref{first-rho-condition} using
$ t = r^{2} $ getting the 
equivalent conditions 
\begin{equation}
\label{rho-condition}
\int_{0}^{R_{w}^{2}} d t \, \rho (t^{1/2}) \, t^{j} = 
|q|^{-j(j+1)} \, w_{j}/\pi 
\quad \mathrm{for~all~} j \ge 0. 
\end{equation}
To find a Borel measure on an 
interval with prescribed moments
as in \eqref{rho-condition} 
is either a Hausdorff moment problem 
if $ 0 < R_{w} < \infty$ or a Stieljes 
moment problem if $ R_{w} = \infty $. 
These problems have been well studied 
in classical analysis 
as to the existence and uniqueness 
of solutions.  
(See the text \cite{schmudgen}.)
For example, taking the case $ w_{j} = j! $ 
and $ |q| =1 $ 
(for which $ R_{w} = \infty $ as we have already seen), 
we can then take $ \rho(r) = \pi^{-1} \exp(-r^{2}) $ 
or, equivalently $ \rho(t^{1/2}) = \pi^{-1} e^{-t} $, 
by standard identities from integral calculus. 
See \cite{dey} for other examples in a 
non-commutative setting. 
Of course, after successfully finding a function 
$ \rho $ solving the moment problem  
\eqref{first-rho-condition}, one still has to check 
that the possibly stronger condition 
\eqref{i-p-resolution} holds. 
\end{remark}

Also, Equation~\eqref{resolution-of-id}, if it holds, 
allows us to identify readily all vectors
$ \psi \in \mathcal{H} $ that are orthogonal 
to all of the coherent states, which means that  
$ \langle \phi_{\lambda}, \psi \rangle = 0 $ 
for all $ \lambda \in B_{w} $. 
One simply uses the equivalent equation 
\eqref{i-p-resolution}
to see that $ \langle \phi , \psi \rangle =0 $
for all $ \phi \in \mathcal{H} $, implying 
 $ \psi = 0 $. 
Another way of saying this is that 
$ \mathrm{span} 
\{ \phi_{\lambda} ~|~ \lambda \in B_{w} \} $ 
is a dense subspace  
of $ \mathcal{H} $ provided that 
\eqref{resolution-of-id} holds.

\section{Time Evolution}

The degree of the homogeneous elements 
can be used to define an operator $ N $
that for all integers $ n \ge 0 $ satisfies
$$
   N \phi_{n} = \deg (\phi_{n}) \, \phi_{n}
              = n \, \phi_{n}. 
$$
By standard techniques in functional analysis 
this has a unique extension as an unbounded, 
self-adjoint linear operator acting in 
a dense domain of the Hilbert space
$ \mathcal{H} $. 
One says that $ N $ is the {\em number operator}. 
It also serves as a quantum Hamiltonian. 
The time evolution unitary group generated by $ N $
is $ e^{- i t N} $, where $ t \in \mathbb{R} $ is 
interpreted as the (dimensionless!) time. 
Then we have immediately by the functional calculus 
that
$$
    e^{- i t N} \, \phi_{n} = e^{- i t n} \, \phi_{n}
    \quad \mathrm{for~all~integers~} n \ge 0. 
$$
Consequently, the time evolution of a coherent state 
$ \phi_{\lambda} $ for $ \lambda \in B_{w} $ 
is given by
\begin{align*}
 e^{- i t N} \, \phi_{\lambda} &= 
 \sum_{n=0}^{\infty} \lambda^{n} \, q^{n(n+1)/2} \, 
  w_{n}^{-1/2} 
  e^{- i t N} \, \phi_{n}
\\
&= \sum_{n=0}^{\infty} \lambda^{n} \, q^{n(n+1)/2} \,
w_{n}^{-1/2} 
e^{- i t n} \, \phi_{n}
\\
&=
\sum_{n=0}^{\infty} (\lambda e^{- i t })^{n} 
\, q^{n(n+1)/2} \, w_{n}^{-1/2}  \, \phi_{n}
\\
&=
\phi_{\lambda e^{- i t }}.
\end{align*}
This shows that coherent states evolve under the 
flow of the unitary group $ e^{- i t N} $ 
into coherent states. 
Also the induced evolution in the phase space
in time $ t $ 
from the initial condition $ \lambda $
is $ \lambda \mapsto \lambda e^{- i t }$,
which is motion in a circle of radius $ |\lambda| $ 
centered at the origin. 
So the operator  $ N $ generates
a flow in the space of coherent states 
and hence in the phase space.
Notice that the latter flow 
does indeed leave the phase space 
$ B_{w} $ invariant. 

It is curious that the 
(appropriately normalized) 
quantum harmonic oscillator 
is unitarily equivalent to $ N + \frac{1}{2} I $, 
since neither the Manin plane nor the 
annihilation operator $ T_{\overline{\theta}} $ 
have anything to do with it {\em a priori}. 
In fact, the phase space of the 
quantum harmonic oscillator is 
the entire complex plane $ \mathbb{C} $, while 
for many choices of the weights the phase 
space of the Manin plane is 
$ B_{w} \subsetneq \mathbb{C} $. 
For such choices of the weights, clearly 
the Toeplitz quantization of the Manin plane 
is not equivalent to the 
quantum harmonic oscillator.

\section{Coherent State Transform}
\label{section-ch-transform}

The coherent states can be used to define 
a transform from the Hilbert space $ \mathcal{H} $
to a space of anti-holomorphic functions 
whose common domain is the interior of 
the phase space. 
For this we introduce the notations 
$ D_{w}:= \{ \lambda \in \mathbb{C}~|~ 
|\lambda| < R_{w} \} $ and 
the complex vector space 
$$
  \mathcal{A} (D_{w}):= 
  \{ f: D_{w} \to \mathbb{C} ~|~ 
  f \mathrm{~is~anti-holomorphic} \}.
$$
We recall that the open ball $ D_{w} $ is 
a subset of the phase space $ B_{w} $. 
We have the following standard definition. 

\begin{definition}
For $ \psi \in \mathcal{H} $ we 
define $ C \psi $, the 
   {\em coherent state transform of $ \psi $}, by 
$$
C \psi (\lambda):= 
\langle \phi_{\lambda} , \psi \rangle_{\mathcal{H}} 
\quad \mathrm{for~all~} \lambda \in D_{w}. 
$$  
\end{definition}

\begin{remark}
	So, $ C\psi : D_{w} \to \mathbb{C} $. 
	Also, putting $ \psi = \sum_{k} c_{k} \phi_{k} $
	with $ \sum_{k} |c_{k}|^{2} < \infty $, 
	we see that
	$$
	C \psi (\lambda) = 
	\langle \phi_{\lambda} , \psi \rangle_{\mathcal{H}}  =
	\sum_{k} \lambda^{* k} \, q^{*k(k+1)/2} \,
	w_{k}^{-1/2} c_{k},
	$$
	a power series 
	which is clearly anti-holomorphic in
	$ \lambda \in D_{w} $.
	Hence, $ C \psi \in \mathcal{A}(D_{w}) $. 
	Moreover, the mapping 
	$ C : \mathcal{H} \to \mathcal{A}(D_{w}) $ 
    given by $ \psi \mapsto C \psi $ 
	is linear.
\end{remark}

We now would like to find an inner product 
on the range of $ C $ so that 
the coherent state transform $ C $ is unitary. 
Of course, if $ C $ is injective, this can be 
done in a unique and trivial way. 
The point is to find an `intrinsic' definition
of that inner product, that is, a way of defining 
a subspace $ \mathcal{C} $ of 
$ \mathcal{A}(B_{w}) $, then making $ \mathcal{C} $
into a Hilbert space via an inner product on it  
and finally showing that the range 
of $ C $ is $ \mathcal{C} $ and that $ C $ is unitary. 
The next lemma is well known, but bears on the 
present discussion.

\begin{lemma}
	$ C $ is injective if and
	only if 
	$ \mathrm{span} 
	\{ \phi_{\lambda} \,|\, \lambda \in D_{w} \} $ 
	is dense 
	in $ \mathcal{H} $
\end{lemma}
\begin{proof}
	$ \Rightarrow $: 
	Suppose that $ C $ is injective. 
	Take $ \psi \in 
	\left(  \mathrm{span} 
	\{ \phi_{\lambda} ~|~ \lambda \in D_{w} \} \right) ^{\perp} $. 
	It suffices to prove $ \psi = 0 $.
	We have 
	$ C \psi (\lambda) =
	 \langle \phi_{\lambda} , \psi \rangle = 0 $ 
	for all $ \lambda \in D_{w} $. 
	Therefore, $ C \psi = 0 $. 
	And then, by the 
	hypothesis that $ C $ is injective, we see 
	that $ \psi = 0 $. 
	
	$ \Leftarrow $: Suppose that 
	$ \mathrm{span} 
	\{ \phi_{\lambda} ~|~ \lambda \in D_{w} \} $ 
	is dense. 
	Take $ \psi \in \ker \, C $. 
	It suffices to show that $ \psi = 0 $ in 
	order to prove that $ C $ is injective. 
	But then $ 0 = C \psi (\lambda) =
	\langle \phi_{\lambda} , \psi \rangle $ 
	for all $ \lambda \in D_{w} $. 
	It follows that  
	$ \psi \in 
	\left(  \mathrm{span} 
	\{ \phi_{\lambda} ~|~ \lambda \in D_{w} \} \right) ^{\perp} $, 
	which is $ 0 $ by the hypothesis. 
	So we see that $ \psi = 0 $ as desired. 
\end{proof} 

The existence of a resolution of the 
identity is the key for the next result. 

\begin{theorem}
	Suppose that there exists a positive measure 
	$ \rho $ on $ B_{w} $ such that
	the resolution of the identity 
	\eqref{resolution-of-id} holds. 
	Suppose that $ \rho(\partial B_{w}) = 0 $ 
	in the case that the boundary 
	$ \partial B_{w} \subset B_{w} $. 
	Then the coherent state transform $ C $
	is a unitary transform from $ \mathcal{H} $ 
	onto its range $ \mathrm{Ran} \, C $ in 
	$ L^{2} (D_{w}, \rho ) $. 
	Consequently, $ \mathrm{Ran} \, C $ is a closed 
	subspace of 	$ L^{2} (D_{w}, \rho ) $. 
	     
	\vskip 0.1cm \noindent 
	{\bf Notation:} 
	$ C : \mathcal{H} \stackrel{\cong}{\longrightarrow} 
	 \mathcal{SB} \subset 
	 L^{2} (D_{w}, \rho ) $, 
	 where $ \mathcal{SB}:= \mathrm{Ran} \, C $ 
	 is called a {\em generalized Segal-Bargmann space.}
     (See \cite{barg}.)
     
\vskip 0.1cm
Moreover, $ \mathcal{SB} $ is a reproducing kernel 
Hilbert space of anti-holomorphic functions 
with reproducing kernel {\rm function} defined 
for $ \mu, \lambda \in D_{w} $ by 
\begin{equation}
\label{repro-kernel}
K(\mu, \lambda) := 
\langle 
\phi_{\mu} , \phi_{\lambda}
\rangle_{\mathcal{H}} = 
\sum_{n=0}^{\infty} \mu^{* n} \lambda^{n}
\, |q|^{n(n+1)} \,
 w_{n}^{-1}\! . 
\end{equation}
\end{theorem}
\noindent 
\textbf{Remark:} 
The existence of a reproducing kernel Hilbert space
isomorphic to $ \mathcal{H} $ is a nice property, 
but it depends on the existence of a 
resolution of the identity, 
which we will not have in many 
interesting cases. 
And, of course, there could be more than 
one resolution of the identity. 
\begin{proof}
Let $ \psi_{1}, \psi_{2} \in \mathcal{H}$ be given.
Then we compute
\begin{align*}
\langle C \psi_{1}, C \psi_{2} \rangle_{L^{2}(D_{w}, \rho)} &=   
\int_{D_{w}} d \rho (\lambda) 
(C \psi_{1} (\lambda) )^{*} C \psi_{2} (\lambda)
\\
&=   
\int_{B_{w}} d \rho (\lambda) 
\langle 
    \psi_{1} , \phi_{\lambda}  
\rangle_{\mathcal{H}}
\langle 
\phi_{\lambda}  , \psi_{2} 
\rangle_{\mathcal{H}} 
\\
&=   
\langle 
\psi_{1} ,  \psi_{2} 
\rangle_{\mathcal{H}}, 
\end{align*}
where we used the resolution of the identity 
\eqref{i-p-resolution} 
in the last equality. 
This shows exactly that $ C $ is a unitary 
transform onto its range. 
In particular, $ C $ is injective. 

Next, let's show that the function in
\eqref{repro-kernel} has the reproducing property. 
So we take an arbitrary element 
$ f \in  \mathcal{SB}$. 
Note that $ f = C \psi  $ for a unique element 
$ \psi \in \mathcal{H} $. 
Then we calculate 
\begin{align*}
\int_{D_{w}} \!\!\! d  \rho (\lambda) \, 
K(\mu, \lambda) f (\lambda) &=
\int_{B_{w}} \!\!\! d \rho (\lambda) \, 
\langle 
\phi_{\mu} , \phi_{\lambda}
\rangle 
C\psi (\lambda)
\\
&=
\int_{B_{w}} \!\!\! d \rho (\lambda) \, 
\langle 
\phi_{\mu} , \phi_{\lambda}
\rangle 
\langle \phi_{\lambda}, \psi \rangle 
\\
&= 
\langle 
\phi_{\mu} , \psi \rangle 
\\
&= 
C \psi (\mu)
\\
&= 
f (\mu). 
\end{align*}
We used the hypothesis $ \rho(\partial B_{w}) = 0 $ 
in the first equality. 
This shows the reproducing property. 

But we also have to show that 
$ K(\cdot, \lambda) \in \mathcal{SB}$ for 
every $ \lambda \in D_{w} $, this being the second 
defining property of a reproducing kernel function. 
But for all $ \mu \in D_{w} $ we have that 
\begin{equation}
\label{kernel-equals-etc}
K (\mu, \lambda) = 
\langle \phi_{\mu} , \phi_{\lambda} \rangle =
C \phi_{\lambda} (\mu) 
\end{equation}
so that $ K (\cdot, \lambda) = C \phi_{\lambda} 
\in \mathrm{Ran} \, C = \mathcal{SB} $ as desired. 

So the function 
$ K: D_{w} \times D_{w} \to \mathbb{C} $ 
in \eqref{repro-kernel} satisfies 
the two defining properties of the 
(unique, it it exists) 
reproducing kernel function for $ \mathcal{SB} $. 
\end{proof}

There is an immediate corollary to this proof. 
\begin{corollary}
	Assume the hypothesis of the previous theorem.
	Then the coherent state transform of a generic 
	coherent state equals the 
	reproducing kernel of the generalized 
	Segal-Bargmann space. 
\end{corollary}
\begin{proof}
We just read \eqref{kernel-equals-etc} backwards, 
namely
$
C \phi_{\lambda} (\mu)  = K (\mu, \lambda)
$.
This equality is what the corollary says 
clumsily in words.
\end{proof}
\begin{remark}
	The property in this corollary already appears in 
	Bargmann's 1961 seminal paper \cite{barg}. 
	It is one of the characteristic properties of
	coherent states, though it is not always 
	mentioned. 
\end{remark}

We can re-write one result of this theorem as 
$  \mathcal{SB} \subset 
L^{2}_{AH} (D_{w}, \rho ) $, the subspace 
of $ L^{2} (D_{w}, \rho ) $ of 
anti-holomorphic functions, namely, 
$$
L^{2}_{AH} (D_{w}, \rho ):= 
L^{2} (D_{w}, \rho ) \cap \mathcal{A} (D_{w}). 
$$
Another approach, more in line with \cite{barg},  
would be to define the Segal-Bargmann 
space as $ L^{2}_{AH} (D_{w}, \rho ) $, which 
as far as we have shown at this point could be 
strictly larger than $ \mathcal{SB} $. 
We would like to show that 
these spaces are actually equal. 
The following is a partial result in that direction. 

\begin{theorem}
\label{SB-equals-anti-L-2-space}
Suppose that the resolution of the identity 
\eqref{resolution-of-id} holds for a measure 
that is radial and absolutely continuous with 
respect to Lebesgue measure. 
Then $ \mathcal{SB} = 
L^{2}_{AH} (D_{w}, \rho ) $. 
\end{theorem}
\begin{proof}
Since $ \mathcal{SB}= \mathrm{Ran} \, C $, it suffices 
to show that 
$ \mathrm{Ran} \, C = L^{2}_{AH} (D_{w}, \rho ) $. 
We have shown already one inclusion. 
So it remains to show 
that any $ f \in L^{2}_{AH} (D_{w}, \rho ) $ lies 
in $ \mathrm{Ran} \, C $. 
Considering such a function $ f $, we can express it 
as 
\begin{equation}
\label{anti-f-expanded}
 f (\lambda) = \sum_{n=0}^{\infty} a_{n} \lambda^{*n} 
 \quad \mathrm{for~all~}  \lambda \in D_{w}
\end{equation}
for certain 
coefficients $ a_{n} \in \mathbb{C} $. 
We let $ 0 < s < R_{w} $ and put 
$ B_{s}  := 
\{ \lambda \in \mathbb{C}~|~ |\lambda| \le s \} $,
the closed ball of radius $ s $. 
The condition that $ f \in L^{2} (D_{w}, \rho ) $ 
implies that the first integral in the 
following calculation is finite: 
\begin{align*}
\int_{B_{s}} d \rho (\lambda) \, |f(\lambda)|^{2} &= 
\int_{B_{s}} d \rho (\lambda) 
\sum_{j,k=0}^{\infty} 
a_{j}^{*} \lambda^{j} a_{k} \lambda^{*k}
\\
&= 
\sum_{j,k=0}^{\infty} a_{j}^{*} a_{k}
\int_{B_{s}} d \rho (\lambda) \,
\lambda^{j} \lambda^{*k}
\\
&= 
\sum_{j=0}^{\infty} |a_{j}|^{2} \,
( 2 \pi )\int_{0}^{s} dr \, \rho(r) \, r^{2 j + 1}. 
\end{align*}
The interchange of the integral and the infinite sum 
in the second equality is justified 
since the double series 
converges uniformly on $ D_{s} $ to its limit 
by standard properties of power series  
and by the hypothesis that $ \rho $ is absolutely 
continuous with respect to Lebesgue measure. 
Taking the limit as $ s \to R_{w} $ we get 
\begin{align*}
\int_{D_{w}} d \rho (\lambda) \, |f(\lambda)|^{2} 
&= 
\lim_{s \to R_{w}} 
\int_{B_{s}} d \rho (\lambda) \, |f(\lambda)|^{2} 
\\
&= 
\lim_{s \to R_{w}}  
\sum_{j=0}^{\infty} |a_{j}|^{2} \,
( 2 \pi ) \int_{0}^{s} dr \, \rho(r) \, r^{2 j + 1} 
\\
&=
\sum_{j=0}^{\infty} |a_{j}|^{2} \,
( 2 \pi ) 
\lim_{s \to R_{w}}
\int_{0}^{s} dr \, \rho(r) \, r^{2 j + 1} 
\\
&=
\sum_{j=0}^{\infty} |a_{j}|^{2} \,
( 2 \pi ) 
\int_{0}^{R_{w}} dr \, \rho(r) \, r^{2 j + 1} 
\\
&=
\sum_{j=0}^{\infty} |a_{j}|^{2} |q|^{-j(j+1)} w_{j}. 
\end{align*}
Here we used the Lebesgue
monotone convergence theorem
in the first, third and fourth equalities. 
The last equality follows from 
\eqref{first-rho-condition}. 
We have shown that 
\begin{equation}
\label{norm-squared-identity}
|| f ||^{2} = \sum_{j=0}^{\infty} 
|a_{j}|^{2} |q|^{-j(j+1)} w_{j} < \infty
\end{equation}
for all $ f \in L^{2}_{AH} (D_{w}, \rho ) $, 
where $ f $ is given by the series in 
\eqref{anti-f-expanded}. 

We now consider the anti-holomorphic 
monomials $ \lambda^{*j} $ 
for $ j \ge 0 $.  
Then one proves that 
$ || \lambda^{*j} ||^{2} = 
|q|^{-j(j+1)} w_{j} < \infty$ 
by evaluating the integral, and 
so we have that 
$ \lambda^{*j} \in  L^{2}_{AH} (D_{w}, \rho )$. 
Consequently, all the anti-holomorphic polynomials are 
in $ L^{2}_{AH} (D_{w}, \rho ) $. 
Another easy calculation shows that
$ \langle \lambda^{*j} , \lambda^{*k} \rangle = 0 $ 
for all $ j \ne k $. 
This combines to show that 
$$ 
\mathcal{B_{AH}}:= 
\{ q^{*j(j+1)/2} w_{j}^{-1/2} \lambda^{*j} \}
$$ 
is an orthonormal {\em set} in the Hilbert space 
$ L^{2}_{AH} (D_{w}, \rho ) $. 

Next we claim that any 
$ f \in L^{2}_{AH} (D_{w}, \rho ) $ 
as in \eqref{anti-f-expanded} is the 
limit in the $ L^{2} $-norm topology of the sequence 
of its partial sums 
$ f_{N}(\lambda) := 
\sum_{j=0}^{N} a_{j} \lambda^{*j} $, which 
are anti-holomorphic polynomials. 
This is so since 
$ f - f_{N} \in L^{2}_{AH} (D_{w}, \rho ) $ 
and therefore by \eqref{norm-squared-identity} applied 
now to $ f - f_{N} $ we have that 
$$
|| f - f_{N} ||^{2} = 
\sum_{j=N+1}^{\infty} |a_{j}|^{2}  |q|^{-j(j+1)} w_{j},
$$ 
which goes to $ 0 $ as $ N \to \infty $ 
by \eqref{norm-squared-identity} applied to $ f $. 
The whole point of the proof so far is that 
we can conclude from this that $ \mathcal{B}_{AH} $ 
is an orthonormal {\em basis} of the Hilbert space 
$ L^{2}_{AH} (D_{w}, \rho ) $.

We next calculate the coherent state transforms 
of the standard basis elements 
$\phi_{j} \in \mathcal{H} $. 
For every integer $ j \ge 0 $ 
and $ \lambda \in D_{w} $ this gives us 
$$
C \phi_{j} (\lambda)
= 
\langle \phi_{\lambda} , \phi_{j} \rangle 
=  
\Big\langle 
\sum_{n=0}^{\infty} \lambda^{n} q^{n(n+1)/2} 
 w_{n}^{-1/2} \phi_{n}, 
\phi_{j} 
\Big\rangle 
=
q^{*j(j+1)/2} w_{j}^{-1/2} \lambda^{*j}. 
$$
Therefore, $ C $ maps the orthonormal 
basis $ \{ \phi_{j} \,|\, j \ge 0\} $ 
of $ \mathcal{H} $ 
onto the  orthonormal basis $ \mathcal{B}_{AH} $
of $ L^{2}_{AH} (D_{w}, \rho ) $, 
which proves that 
$ \mathrm{Ran} \, C = L^{2}_{AH} (D_{w}, \rho ) $ 
as desired. 
\end{proof}

\begin{remark}
In this argument we proved in \eqref{norm-squared-identity} that for any 
anti-holomorphic function $ f \in L^{2} (\rho)$ 
as in \eqref{anti-f-expanded}, 
we have 
$ \sum_{j} |a_{j}|^{2} |q|^{j(j+1)} w_{j} $ 
is finite and is 
equal to $ ||f ||^{2}_{L^{2} (\rho)} $. 
We did not prove the converse, which we now do. 
\end{remark}

\begin{theorem}
Assume the same hypothesis as in the previous 
theorem. 
Suppose an anti-holomorphic function 
$ f $ is given as in 
\eqref{anti-f-expanded} and that 
$ \sum_{j} |a_{j}|^{2} |q|^{j(j+1)} w_{j} < \infty$. 
Then $ f \in L^{2}_{AH} (D_{w}, \rho ) $ and 
$ || f ||^{2}_{L^{2} (\rho)} = 
\sum_{j} |a_{j}|^{2}  |q|^{j(j+1)} w_{j}  $. 
\end{theorem}
\begin{proof}
We consider the partial sums
$f_{N}(\lambda) = \sum_{j=0}^{N} a_{j} \lambda^{*j}$, 
which by the previous theorem are in 
$ L^{2}_{AH} (D_{w}, \rho ) $ 
and satisfy 
$|| f_{N} ||^{2} = 
\sum_{j=0}^{N} |a_{j}|^{2}  |q|^{j(j+1)} w_{j}$. 
Moreover, by \eqref{anti-f-expanded} we have
$$
\lim_{N \to \infty} f_{N} (\lambda) = f (\lambda) 
\quad \mathrm{for~all~}\lambda \in D_{w}. 
$$
So, by the Fatou lemma we estimate 
\begin{align*}
\int_{D_{w}} d \rho (\lambda) \, | f(\lambda)|^{2} &= 
\int_{D_{w}} d \rho (\lambda) \, \lim_{N \to \infty} 
|  f_{N}(\lambda) |^{2}
\\
&\le  
\liminf_{N \to \infty} 
\int_{D_{w}} d \rho (\lambda) \,  
|  f_{N}(\lambda) |^{2}
\\
&=
\liminf_{N \to \infty} || f_{N} ||^{2}
\\
&=  
\lim_{N \to \infty} \sum_{j=0}^{N} 
|a_{j}|^{2}  |q|^{j(j+1)} w_{j} 
\\
&= 
\sum_{j=0}^{\infty} |a_{j}|^{2} 
 |q|^{j(j+1)} w_{j} < \infty. 
\end{align*}
This proves that $ f \in L^{2}_{AH} (D_{w}, \rho ) $ 
but only gives an estimate on its $ L^{2} $-norm. 
Next a similar argument using the Fatou lemma shows 
for every integer $ N \ge 0 $ that 
$$
|| f - f_{N} ||^{2} \le \sum_{j=N+1}^{\infty} 
|a_{j}|^{2} |q|^{j(j+1)} w_{j}, 
$$
which goes to $ 0 $ as $ N \to \infty $. 
This says that $ f_{N} \to f $ in the topology of the 
$ L^{2} $-norm. 
This in turn implies by the continuity of  the 
$ L^{2} $-norm that  
$$
|| f ||^{2} = \lim_{N \to \infty} || f_{N} ||^{2}
\\
= \lim_{N \to \infty} \sum_{j=0}^{N} |a_{j}|^{2} w_{j} 
= \sum_{j=0}^{\infty} |a_{j}|^{2} w_{j}. 
$$
And that finishes the proof. 
\end{proof}

As far as this analysis goes 
it still remains a logical 
possibility that the inclusion 
$  \mathcal{SB} \subset 
L^{2}_{AH} (D_{w}, \rho ) $ is proper for
other measures $ \rho $. 
What is happening for such measures is an 
open problem.

\section{Another Toeplitz Quantization}

We continue to assume  
\eqref{resolution-of-id} holds for 
some measure  
$ \rho $ on $  B_{w} $ in this section. 
We proved in the last section that 
the Segal-Bargmann space $ \mathcal{SB} $ 
is a closed subspace of $ L^{2} (D_{w}, \rho) $
and that it is a reproducing kernel Hilbert space. 
This gives the standard set-up for defining 
Toeplitz operators whose symbols are functions. 
First one uses the kernel function 
\eqref{repro-kernel} as the kernel of an 
integral operator that defines an 
orthogonal projection 
$ P_{K}: L^{2} (D_{w}, \rho) \to L^{2} (D_{w}, \rho) $ 
by 
$$
(P_{K} f )(\mu) := 
\int_{D_{w}} d \rho(\lambda) K(\mu,\lambda) f (\lambda) 
\quad \mathrm{for~}f \in L^{2} (D_{w}, \rho)  
\mathrm{~and~}\mu \in D_{w}. 
$$
This integral converges absolutely for 
$ \mu \in D_{w} $ by the 
Cauchy-Schwarz inequality together with the fact that
$ K(\mu, \cdot) \in L^{2} (D_{w}, \rho) $. 
By a standard argument $ (P_{K} f )(\mu) $
is anti-holomorphic in $ \mu \in D_{w} $. 
Moreover, by the resolution of the identity 
\eqref{resolution-of-id}, $ P_{K} $ acts as the 
identity on $ \mathcal{SB} $. 
However, we only have shown that 
$ P_{K} f \in L_{AH}^{2} (D_{w}, \rho) $, 
which as we noted earlier might be strictly 
larger than $ \mathcal{SB} $ 
for some measures $ \rho $. 
It even seems possible to have
$$
\mathcal{SB} \subsetneq \mathrm{Ran} \, P_{K} 
\subsetneq L_{AH}^{2} (D_{w}, \rho). 
$$
To avoid such details  
for the rest of this section we assume
that $\mathcal{SB} = L_{AH}^{2} (D_{w}, \rho) $, 
which we know holds in many cases according 
to Theorem~\ref{SB-equals-anti-L-2-space}. 
Also $ \mathcal{SB} $ has a lot of nice 
structure, such as an explicit reproducing 
kernel function and a standard orthonormal basis,
making it a more preferable domain for 
Toeplitz operators. 

\begin{definition}
For $ f \in L^{\infty} (D_{w}, \rho) $ 
we define the  
{\em (secondary) Toeplitz operator with symbol $ f $} 
by 
$$
    S_{f} \, \phi := P_{K} (f \phi) 
    \quad \mathrm{for~all~} \phi \in \mathcal{SB}. 
$$
\end{definition}
\begin{remark}
The notation $ S_{f} $ distinguishes this from 
the Toeplitz operators $ T_{g} $ which we introduced in 
Section~\ref{setting-section}.  
Clearly, $ S_{f} : \mathcal{SB} \to \mathcal{SB} $ 
is linear. 
Moreover, we call the mapping 
$$ 
S : L^{\infty} (D_{w}, \rho) \to \
\mathcal{L} (\mathcal{SB}) := 
\{ T: \mathcal{SB} \to \mathcal{SB} ~|~ 
T \mathrm{~is~linear}\}
$$ 
given by $ f \mapsto S_{f} $ 
the {\em (secondary) Toeplitz quantization} of 
the commutative algebra
$ L^{\infty} (D_{w}, \rho) $. 
Also $ S $ is linear. 
Recall that $ D_{w} $ is the interior of the 
phase space associated to the Manin plane. 
So this is another quantization scheme associated 
with the Manin plane. 
Many standard properties hold for this theory. 
But we will leave the development of them 
for future research. 
\end{remark}

\section{Coherent State Quantization}
\label{section-cs-quantization}

A general reference for 
coherent state quantization is Part~II of
the text \cite{gazeau} as well as
the papers  \cite{cotfas1} and \cite{cotfas2}. 
For an abstract approach see \cite{ali-englis}. 
To start off this section we assume  
\eqref{resolution-of-id} holds for 
some measure  
$ \rho $ on $  B_{w} $. 
This is the key property in order to be able
to define the coherent state quantization. 

\begin{definition}
Suppose $ f : B_{w} \to \mathbb{C} $ is a 
measurable function. 
We then define the 
{\em coherent state quantization of $ f $} to be the linear operator
\begin{equation}
\label{formal-Q-cs-definition}
Q_{cs} f := \int_{B_{w}} 
d \rho(\lambda) \,  f(\lambda) \, 
| \phi_{\lambda} \rangle \langle \phi_{\lambda} |, 
\end{equation}
where 
$ | \phi_{\lambda} \rangle \langle \phi_{\lambda} | $ 
is Dirac notation for the rank one 
projection operator given by
$ \psi \mapsto \langle \phi_{\lambda}, \psi \rangle  \phi_{\lambda} $ 
for all $ \psi \in \mathcal{H} $. 
This is also called the 
{\em frame quantization of $ f $}. 
The integral in 
\eqref{formal-Q-cs-definition}
is understood to mean the unique 
linear operator 
$ Q_{cs} f : \mathcal{H} \to \mathcal{H}$ 
(if it exists) which satisfies 
\begin{equation}
\label{define-Q-cs}
\langle \phi, (Q_{cs} f) \psi \rangle = 
\int_{B_{w}} d  \rho(\lambda) \,  f(\lambda) \, 
\langle \phi,  \phi_{\lambda} \rangle 
\langle \phi_{\lambda} , \psi \rangle 
\end{equation}
for all $ \phi, \psi \in \mathcal{H} $, where 
the right hand side of \eqref{define-Q-cs} 
is the Lebesgue integral 
of a complex valued, integrable function. 
	
Also, define 
$ \mathcal{L}_{cs}^{1} (B_{w}, \rho) $ to be 
the set of all those measurable functions
$ f : B_{w} \to \mathbb{C} $ for which the integral
in \eqref{define-Q-cs} exists for all 
$ \phi, \psi \in \mathcal{H} $, 
that is to say, 
its integrand is absolutely integrable. 
\end{definition}

\begin{remark}
\label{remark-section-8}
	The function 
	$ \langle \phi_{\lambda} , \psi \rangle $
	(resp., 
	$ \langle \phi,  \phi_{\lambda} \rangle $) 
	of $ \lambda \in D_{w} $ is anti-holomorphic 
	(resp., holomorphic) 
	and therefore is a measurable function. 
	Thus, the integrand in 
	\eqref{define-Q-cs} is a measurable function 
	of $ \lambda \in D_{w} $. 
	Also, it is clear that 
	$ \mathcal{L}_{cs}^{1} (B_{w}, \rho) $ 
	is a complex vector space with respect to the 
	standard point-wise definitions of sum and 
	scalar product. 
    Another immediate property is 
    $ Q_{cs} (f^{*}) = ( Q_{cs} (f) )^{*} $, 
    the adjoint operator of $ Q_{cs} (f) $. 
    And the resolution of the identity 
    \eqref{resolution-of-id} tells us that
    $ Q_{cs} 1 = I $, where $ 1 $ denotes the 
    constant function. 
\end{remark}

At this point let us recall some standard notations. 
We let $ \mathcal{B} (\mathcal{H})$ denote the 
Banach space of all linear, bounded maps 
$ A: \mathcal{H} \to \mathcal{H} $ where the 
norm $ || A || $ of such an $ A $ is its operator norm 
$|| A || := \sup \{ || A \psi || ~|~ 
\psi \in \mathcal{H}, \, 
||\psi|| \le 1\}$. 

\begin{prop}
	Suppose $ f : B_{w} \to \mathbb{C} $ 
	is integrable with respect to the 
	measure 
	$ || \phi_{\lambda} ||^{2} \, d \rho 
	 =   
	d \rho(\lambda) \, 
	|| \phi_{\lambda} ||^{2} $ on 
	$ B_{w} $, that is, 
	\begin{equation}
	\label{f-in-L1}
	f \in L^{1} \big(B_{w}, 
	 || \phi_{\lambda} ||^{2} \, d \rho \big).
	\end{equation}
    Then $ f \in \mathcal{L}_{cs}^{1} (B_{w}, \rho) $
	and $ Q_{cs} f : \mathcal{H} \to \mathcal{H} $ 
	is a bounded operator whose operator norm 
	satisfies 
	$ || Q_{cs} f || \le || f ||_{1}$,  
	the $ L^{1} $-norm of $ f $ in the $ L^{1} $-space 
	in \eqref{f-in-L1}. 
    Therefore, $ Q_{cs} :  
    L^{1} \big(B_{w}, 
    || \phi_{\lambda} ||^{2} \, d \rho \big) 
    \to \mathcal{B} (\mathcal{H})$ is a bounded 
    linear map of Banach spaces with 
    operator norm $ || Q_{cs} || \le 1 $. 
\end{prop}
\begin{proof}
We prove the absolute integrability of the  
integral in \eqref{define-Q-cs} 
for all $ \phi, \psi \in \mathcal{H} $
as follows: 
\begin{align}
\label{Q-cs-f-estimate}
\int_{B_{w}} d 
\rho(\lambda) \, \big| f(\lambda) \,& 
\langle \phi,  \phi_{\lambda} \rangle 
\langle \phi_{\lambda} , \psi \rangle \big| 
\le 
\left( 
 \int_{B_{w}} d \rho(\lambda) \, 
 || \phi_{\lambda}  ||^{2} \, \big| f(\lambda) \big| \,
\right) 
 \, || \phi || \, || \psi || 
\\
&= 
 || f ||_{1} \, || \phi || \, || \psi || < \infty,
 \nonumber
\end{align}
where we used the Cauchy-Schwarz inequality twice 
and the definition of $ || f ||_{1} $. 
Referring back to \eqref{define-Q-cs} we see that
$ |\langle \phi, (Q_{cs} f) \psi \rangle| 
\le  || f ||_{1} \, || \phi || \, || \psi ||$ 
for all $ \phi, \psi \in \mathcal{H} $, from which 
the bound on the operator norm 
$ || Q_{cs} f || \le  || f ||_{1} $ follows directly. 
Then we see immediately that $ || Q_{cs} || \le 1 $. 
\end{proof}

\begin{remark}
We will not elaborate on the standard details 
needed to extend this definition to 
(possibly unbounded)
operators $ Q_{cs} f $ for 
$ f $ in other spaces, including spaces of 
distributions. 
The condition \eqref{f-in-L1} implies that $ f $ 
has some decay that cancels the divergence of 
the integral \eqref{divergent-integral}. 
In particular the constant function $ f \equiv 1 $ 
does not satisfy \eqref{f-in-L1}. 
Nonetheless, the resolution of the identity says that 
$ Q_{cs} 1 = I$, the identity map.  
So the condition \eqref{f-in-L1} is not necessary 
for $ Q_{cs} f $ to be bounded; it only is sufficient. 
We leave finding a nice necessary and sufficient 
condition as an open problem 
\end{remark}

\section{Upper and Lower Symbols}

We now discuss a standard topic in the theory 
of coherent states. 
This dates back to the seminal works 
of Berezin in \cite{berezin1} and \cite{berezin2}, 
Glauber in \cite{glauber1} and \cite{glauber2} 
and Lieb in \cite{lieb}.    

\begin{definition}
\label{define-lower-symbol}
Let $ A $ be a densely defined linear operator
acting in the Hilbert space $ \mathcal{H} $. 
Suppose that for each $ \lambda \in B_{w} $ 
we have that 
$ \phi_{\lambda} \in \mathrm{D} (A) $, 
the domain of $ A $. 
Then the {\em unnormalized lower symbol} 
$ A^{\sharp} $
of $ A $ is defined for all $ \lambda \in B_{w} $ by 
$$
A^{\sharp}(\lambda):= 
\langle \phi_{\lambda}, A\phi_{\lambda} \rangle.
$$
One also says that $ A^{\sharp} $ is the 
{\em unnormalized covariant symbol} of $ A $. 

Similarly, the {\em (normalized) lower  symbol} 
$ A^{\flat} $ is defined by 
$$
A^{\flat}(\lambda):= 
\dfrac{\langle \phi_{\lambda}, A\phi_{\lambda} \rangle}
     {\langle \phi_{\lambda}, \phi_{\lambda} \rangle}
=
   \dfrac{A^{\sharp} (\lambda)}
   {\langle \phi_{\lambda}, \phi_{\lambda} \rangle}, 
$$
again for all $ \lambda \in B_{w} $.
(Recall that $ \phi_{\lambda} \ne 0 $ so that the 
denominator is non-zero.) 
Also one says $ A^{\flat} $ is the 
{\em (normalized)  covariant symbol} of $ A $ or the
{\em Berezin symbol}. 
\end{definition}
\begin{remark}
Therefore the lower symbol 
$ A^{\flat} : B_{w} \to \mathbb{C} $, 
that is, $ A^{\flat} $ is a function 
on the phase space. 
Also, if $ A $ is self-adjoint, which says
that it
represents a quantum observable, then 
$ A^{\flat} : B_{w} \to \mathbb{R} $, which 
says that it represents a classical observable. 
If $ A $ is a bounded operator, then 
$ A^{\flat} $ is a bounded function 
satisfying $ || A^{\flat} ||_{\infty} \le || A || $, 
where $ || \cdot ||_{\infty} $ is 
the $ L^{\infty} $ norm of a bounded function. 
An elementary  
example of a lower symbol is given by 
$ (T_{1})^{\flat} (\lambda) = I^{\flat} (\lambda) =1$, 
the constant function. 
In quantum theory
quantization is a (certain!)
way of passing from 
functions on phase space to operators acting in 
a Hilbert space. 
Since the operation $ A \mapsto A^{\flat} $ 
is a mapping in the opposite direction
(namely, from operators to functions on phase space), 
one sometimes refers to it 
as a {\em dequantization}.
Analogous comments hold for the unnormalized lower 
symbol $ A^{\sharp} $. 
\end{remark}

Also, the unnormalized lower symbol is related to 
the coherent state transform and the reproducing 
kernel by
$$
A^{\sharp} (\lambda) = 
\langle \phi_{\lambda}, A\phi_{\lambda} \rangle = 
C(\phi_{\lambda}) (\lambda) = 
K(\lambda,\lambda) \quad \mathrm{for~all~} 
\lambda \in D_{w}. 
$$

The following property for 
self-adjoint operators is well known. 
We wish merely to emphasize that it comes from 
a property for  
a wider class of operators.
Before stating this result, we recall that the
adjoint of a densely defined operator $ A $ acting 
in a Hilbert space is denoted as $ A^{*} $. 

\begin{prop}
\label{lower-symbol-A-star}
Let $ A $ be an operator as in 
Definition~\ref{define-lower-symbol}. 
Suppose that for each $ \lambda \in B_{w} $ 
we have that 
$ \phi_{\lambda} \in \mathrm{D} (A^{*}) $, 
the domain of $ A^{*} $. 
(This latter hypothesis guarantees that 
$ (A^{*})^{\sharp} $ and $ (A^{*})^{\flat} $
are defined.
For example, it holds for symmetric operators.) 
Then $ (A^{*})^{\sharp} = (A^{\sharp})^{*} $ and 
$ (A^{*})^{\flat} = (A^{\flat})^{*} $. 

In particular, if $ A $ is self-adjoint, then 
both $ A^{\sharp} $ and $ A^{\flat} $ are 
real-valued functions. 
\end{prop}
\begin{proof}
For all $ \lambda \in B_{w} $ we see that 
$$
(A^{*})^{\sharp}(\lambda)= 
\langle \phi_{\lambda}, A^{*}\phi_{\lambda} \rangle 
=
\langle A \phi_{\lambda}, \phi_{\lambda} \rangle 
=
\langle \phi_{\lambda}, A\phi_{\lambda} \rangle^{*} 
=
(A^{\sharp} (\lambda) )^{*}
=
(A^{\sharp} )^{*}(\lambda). 
$$
The rest of the proposition is now immediate. 
\end{proof}

Now we give new terminology to something we already
have seen. 
\begin{definition}
	Suppose that 
	$ f : B_{w} \to \mathbb{C} $ is 
	 a function for which the
	coherent state quantization $ Q_{cs} f $ exists. 
	Then we say that $ f $ is the {\em upper 
    (or contravariant) symbol} for the 
    operator $ Q_{cs} f $. 
\end{definition}

\begin{remark}
Notice that both the upper and lower symbols 
are complex valued functions defined on $ B_{w} $,
that is, they are classical observables if they 
happen to be real valued.
Moreover, each is associated with a linear operator, 
that is, a quantum observable if it happens to be 
self-adjoint. 
Our presentation of upper symbols is not 
standard. 
Typically one starts with a linear operator $ A $ 
and looks for a function $ f $ such that 
$ A = Q_{cs} f  $, in which case one writes
$ f = \hat{A} $. 
However, such an $ f $ may not exist and, if it does,
it may not be unique.  
\end{remark}

Due to the presence of two quantization schemes 
here, we can ask some questions about their 
relationship. 
For example, what are the lower symbols of
the creation and annihilation operators?   

\begin{theorem}
\label{lower-of-T-theta-bar}
The normalized 
lower symbol of the annihilation operator 
$ T_{\overline{\theta}} $ 
is given by 
$
(T_{\overline{\theta}} )^{\flat} (\lambda) = 
\lambda 
$
for every $ \lambda \in B_{w} $, 
that is, the identity function on $ B_{w} $. 
\end{theorem}

\begin{proof}
For all $ \lambda \in B_{w} $ we see that  
the unnormalized lower symbol is 
\begin{align*}
(T_{\overline{\theta}} &)^{\sharp} (\lambda) = 
\langle \phi_{\lambda} , T_{\overline{\theta}} \, 
\phi_{\lambda} \rangle 
= 
\big\langle 
\phi_{\lambda}, 
\sum_{n=0}^{\infty} \lambda^{n} q^{n(n+1)/2}
w_{n}^{-1/2} 
 T_{\overline{\theta}} \, \phi_{n} 
 \big\rangle 
 \\
 &= 
 \big\langle 
 \phi_{\lambda}, 
 \sum_{n=1}^{\infty} \lambda^{n} q^{n(n+1)/2}
 w_{n}^{-1/2} 
 \left( \dfrac{w_{n}}{w_{n-1}} \right)^{1/2} 
 \!\!\! q^{-n} \phi_{n-1}
\big\rangle 
\\
&= 
\big\langle 
\sum_{k=0}^{\infty} \lambda^{k} q^{k(k+1)/2}
 w_{k}^{-1/2} \phi_{k}, 
\sum_{k=0}^{\infty} \lambda^{k+1} q^{(k+1)(k+2)/2}
 w_{k+1}^{-1/2} 
\left( \dfrac{w_{k+1}}{w_{k}} \right)^{1/2}
\!\!\! q^{-(k+1)} \phi_{k}
\big\rangle 
\\
&= 
\big\langle 
\sum_{k=0}^{\infty} 
\lambda^{k} q^{k(k+1)/2} w_{k}^{-1/2} \phi_{k}, 
\sum_{k=0}^{\infty} 
\lambda^{k+1} q^{k(k+1)/2} w_{k}^{-1/2} \phi_{k}
\big\rangle 
\\
&= 
\lambda \sum_{k=0}^{\infty} |\lambda|^{2k}
|q|^{k (k+1)}
 w_{k}^{-1}. 
\end{align*}
Recall from \eqref{CS-norm-squared} that 
$ || \phi_{\lambda} ||^{2} = 
\sum_{k=0}^{\infty} |\lambda|^{2k} |q|^{k(k+1)}
w_{k}^{-1} $. 
Consequently, the lower symbol of 
the annihilation operator 
$ T_{\overline{\theta}} $ is given by 
the quotient of these expressions, 
namely by 
$
(T^{r}_{\overline{\theta}} )^{\flat} (\lambda) = 
\lambda 
$
\end{proof} 
\begin{remark}
	Notice that the unnormalized lower symbol 
	$ (T_{\overline{\theta}} )^{\sharp} $ 
	depends on both $ q $ and 
	the weights $ w_{k} $. 
	Remarkably, the normalized lower symbol 
	$ (T_{\overline{\theta}} )^{\flat} $ is 
	independent of these parameters.  
\end{remark}

We now use another consequence 
of \eqref{T-theta-i-theta-bar-j}, 
namely that
the creation operator $ T_{\theta} $ is given by  
\begin{equation}
\label{T-theta-formula}
T_{\theta} \, \phi_{n} =  
\left( \dfrac{w_{n+1}}{w_{n}} \right)^{1/2}
\!\!\! \phi_{n+1} \quad \mathrm{for~all~} n \ge 0.
\end{equation}
Unlike the annihilation operator 
$ T_{\overline{\theta}} $ this 
formula does not depend on $ q $. 
Surprisingly, this means that for $ q \ne 1 $ 
the operators $ T_{\theta} $ and 
$ T_{\overline{\theta}} $ are not adjoints of 
each other. 
But this is simply a consequence of our 
definitions.  
Equation \eqref{T-theta-formula} 
defines $ T_{\theta} $ 
on the subspace $ \mathcal{P} $, which is
dense in $ \mathcal{H} $ and 
which is invariant under the action of 
$ T_{\theta} $. 

We now have gathered enough information for
the following calculation for 
the lower symbol of the creation operator.
This is a formal calculation, since we 
do not concern ourselves with domain 
considerations. 
So, we start with the unnormalized 
lower symbol: 
\begin{align*}
(T_{\theta} )^{\sharp} &(\lambda) = 
\langle \phi_{\lambda} , T_{\theta} \, 
\phi_{\lambda} \rangle 
\\
&= 
\big\langle 
\sum_{k=0}^{\infty} \lambda^{k} q^{k(k+1)/2}
w_{k}^{-1/2} \phi_{k}, 
\sum_{n=0}^{\infty} \lambda^{n} q^{n(n+1)/2}
w_{n}^{-1/2} 
T_{\theta} \, \phi_{n} \big\rangle 
\\
&= 
\big\langle 
\sum_{k=1}^{\infty} \lambda^{k} q^{k(k+1)/2}
w_{k}^{-1/2} \phi_{k}, 
\sum_{n=0}^{\infty} \lambda^{n} q^{n(n+1)/2}
w_{n}^{-1/2} 
\left( \dfrac{w_{n+1}}{w_{n}} \right)^{1/2}
\!\!\! \phi_{n+1}
\big\rangle 
\\
&= 
\big\langle 
\sum_{n=0}^{\infty} \lambda^{n+1} q^{(n+2)(n+1)/2}
w_{n+1}^{-1/2} \phi_{n+1}, 
\sum_{n=0}^{\infty} \lambda^{n} q^{n(n+1)/2}
w_{n}^{-1/2} 
\left( \dfrac{w_{n+1}}{w_{n}} \right)^{1/2}
\!\!\! \phi_{n+1}
\big\rangle 
\\
&=
\lambda^{*} \sum_{n=0}^{\infty} 
|\lambda|^{2n} |q|^{n(n+1)} q^{*(n+1)}
w_{n}^{-1}. 
\end{align*}
For $ q \ne 1 $ there is no apparent 
simplification of this formula. 
Also, the formula for the 
normalized lower symbol is equally unattractive. 
This is why we decided not to give the technical 
details to make this calculation rigorous. 
Also, this shows how the Toeplitz operators which use
left multiplication differ from the Toeplitz operators which use right multiplication, since in the latter 
case all factors of $ q $ disappear and everything 
works out too easily. 
For example, for the Toeplitz operators using right 
multiplication, $ T_{\theta} $ and 
$ T_{\overline{\theta}} $ are adjoint of each other
(on the appropriate domain), while that is false 
for the Toeplitz operators using left 
multiplication if $ q \ne 1 $. 

What this suggests is that instead of $ T_{\theta} $ we
should consider 
the adjoint operator 
$ (T_{\overline{\theta}} )^{*} $ to be 
the appropriate creation operator for 
this setting, even though it
might not be a Toeplitz operator. 
An elementary calculation shows that
$$
(T_{\overline{\theta}} )^{*} \phi_{n}
=
q^{-(n+1)} 
\left( 
\dfrac{w_{n+1}}{w_{n}} 
\right)^{1/2} 
\!\!\!\phi_{n+1},
$$
which has degree $ +1 $. 
Then we can extend $ (T_{\overline{\theta}} )^{*} $ 
to $ \mathrm{D} ( T_{\overline{\theta}} )$ by the 
obvious formula. 
Then a simple modification of the previous 
formal calculation shows rigorously 
for all $ \lambda \in B_{w} $ 
and $ q \in \mathbb{R} $ that 
\begin{align*}
(T_{\overline{\theta}} )^{* \sharp} (\lambda) 
&=
\lambda^{*} \sum_{n=0}^{\infty} 
|\lambda|^{2n} |q|^{n(n+1)} q^{*(n+1)} q^{-(n+1)}
w_{n}^{-1}
\\
&=
\lambda^{*} \sum_{n=0}^{\infty} 
|\lambda|^{2n} |q|^{n(n+1)} 
w_{n}^{-1} 
=
\lambda^{*} || \phi_{\lambda} ||^{2}.
\end{align*}
We have proved the next result. 
\begin{theorem}
Suppose that the parameter $ q $ of the Manin plane 
is real. 
Then the lower symbol of 
$ (T_{\overline{\theta}} )^{*} $ is given by 
$ 
(T_{\overline{\theta}} )^{* \flat} (\lambda) = 
\lambda^{*}
$ 
for all $ \lambda \in B_{w} $. 
\end{theorem}

\section{Upper symbol of a Toeplitz operator}

We now consider another relation between these
two quantizations. 
We take a Toeplitz operator and ask 
whether it has the form $ Q_{cs} f$ for 
some upper symbol $ f $.  
If $ f \in L^{1} \big(B_{w}, 
|| \phi_{\lambda} ||^{2} \, d \rho \big) $, 
we have that $ Q_{cs} f$ 
is bounded, and so $ \phi_{n}  $ is in its 
domain.  
When considering the more 
general case of an unbounded
$ Q_{cs} f$ we will only consider the case when 
$ \phi_{n}  $ is in its domain for all $ n \ge 0 $. 
We start off with a calculation
for integers $ k,n \ge 0 $ of a matrix element,   
\begin{align}
\label{Q-cs-f-phi-n}
\langle \phi_{k}&, Q_{cs} ( f ) \phi_{n} \rangle = 
\int_{B_{w}} 
d \rho(\lambda) \,  f(\lambda) \,  
\langle \phi_{k}, \phi_{\lambda} \rangle 
\langle \phi_{\lambda}, \phi_{n} \rangle  
\\
&= 
q^{k(k+1)/2}  w_{k}^{-1/2} q^{*n(n+1)/2} w_{n}^{-1/2} 
\int_{B_{w}} 
d \rho(\lambda) \,  f(\lambda) \, 
\lambda^{k}  \lambda^{*n}. 
\nonumber
\end{align}

One question is whether this 
can give us the annihilation operator 
$ T_{\overline{\theta}}$. 
But  we know from \eqref{T-theta-bar} 
that the matrix elements of 
$ T_{\overline{\theta}}$ for $ n \ge 1 $ are 
$$
\langle 
\phi_{k},  T_{\overline{\theta}} \, \phi_{n} 
\rangle 
=
\delta_{k,n-1} q^{-n}
\left( 
\dfrac{w_{n}}{w_{n-1}}
\right)^{1/2} 
~~~ \mathrm{for~all~} k \ge 0. 
$$
This formula is also valid for $ n = 0 $ 
provided that we put $ w_{-1} = 1 $, say. 
Two operators are equal on $ \mathcal{P} $ 
if and only if their 
matrix elements are equal on $ \mathcal{P} $, 
and so we see 
that $  T_{\overline{\theta}} = Q_{cs}(f) $
on $ \mathcal{P} $  
if and only if the unknown upper symbol 
$ f $ satisfies 
\begin{align}
\label{f-condition}
\int_{B_{w}} 
d \rho(\lambda) \,  &f(\lambda) \, 
\lambda^{k} \lambda^{*n} = 
\\
&=
\delta_{k,n-1} q^{-n}
\Big( 
\dfrac{w_{n}}{w_{n-1}}
\Big)^{1/2} \!\! 
q^{-k(k+1)/2}  w_{k}^{1/2} (q^{*})^{-n(n+1)/2}
w_{n}^{1/2} 
\nonumber
\\
&=
\delta_{k,n-1} q^{-n}    
q^{-n(n-1)/2}  (q^{*})^{-n(n+1)/2}
w_{n} 
\nonumber
\\
&=
\delta_{k,n-1}  
q^{-n(n+1)/2}  (q^{*})^{-n(n+1)/2}
w_{n} 
\nonumber 
\\
&=
\delta_{k,n-1} 
|q|^{-n(n+1)} 
w_{n} 
\nonumber
\end{align}
for all $ k , n \ge 0 $. 
Therefore the question reduces to whether there exists 
an upper symbol $ f $ satisfying 
\eqref{f-condition}. 
We have the following partial answer. 
\begin{theorem}
\label{annihilation-partial-answer}
	Let $ \rho $ be a radial measure 
	that satisfies \eqref{first-rho-condition}.  
	Define 
	$ f : B_{w} \to \mathbb{C} $ by 
	$ f (\lambda) = \lambda $, the identity map. 
	Suppose that 
	$  T_{\overline{\theta}} $ is a bounded operator. 
    Then 
	$  T_{\overline{\theta}} = Q_{cs} f $,
    that is, the operator
    $  T_{\overline{\theta}} $ has $ f $ as 
    an upper symbol. 
\end{theorem}
\begin{proof}
We evaluate the integral in \eqref{f-condition} 
for this choice of $ f $ as follows: 
\begin{align*}
\int_{B_{w}} 
d \rho(\lambda) \,  f(\lambda) \, 
\lambda^{k} \lambda^{*n} 
&=
\int_{B_{w}} 
d \rho(\lambda) \, \lambda \, 
\lambda^{k} \lambda^{*n} 
\\
&= 
\int_{0}^{2 \pi} d \alpha 
\int_{0}^{R_{w}} r \, d r \rho (r) \, r^{k+n+1} 
e^{i (k+1) \alpha} e^{- i n \alpha} 
\\
&=
2 \pi \, \delta_{k+1,n}  
\int_{0}^{R_{w}} d r \rho (r) \, r^{k+n+2} 
\\
&=
2 \pi \, \delta_{k,n-1}  
\int_{0}^{R_{w}} d r \rho (r) \, r^{2n+1} 
\\
&=
2 \pi \, \delta_{k,n-1} \, |q|^{-n(n+1)}
\left( \dfrac{w_{n}}{2 \pi} \right) 
\\
&= 
 \delta_{k,n-1} |q|^{-n(n+1)} w_{n}, 
\end{align*}
where we used \eqref{first-rho-condition} in the 
next to last equality.  
And so \eqref{f-condition} is proved, which shows 
that on $ \mathcal{P} $ we have 
 $ T_{\overline{\theta}} = Q_{cs} f $, 
a bounded operator. 
So, $ T_{\overline{\theta}} = Q_{cs} f $ 
on $ \mathcal{H} $. 
\end{proof}

We also have a partial answer for the 
creation operator $ ( T_{\overline{\theta}} )^{*} $.

\begin{corollary}
	With the same hypotheses as in Theorem 
	\ref{annihilation-partial-answer}
	there exists an upper symbol 
	$ g : B_{w} \to \mathbb{C} $ 
	for $ ( T_{\overline{\theta}} )^{*} $, 
	which is a bounded operator. 
	This means that 
	$  ( T_{\overline{\theta}} )^{*} = Q_{cs} g $. 
\end{corollary}
\begin{proof}
Again taking 
$ f (\lambda) = \lambda $ we have 
$
( T_{\overline{\theta}} )^{*} = 
( Q_{cs} f )^{*} = Q_{cs} (f^{*}),
$
where  we used 
Theorem \ref{annihilation-partial-answer} 
in the first equality 
and a basic property of the coherent state 
quantization $ Q_{cs} $ in the 
second equality. 
So we take 
$ g(\lambda) = f^{*}(\lambda) = \lambda^{*} $. 
The boundedness of $ ( T_{\overline{\theta}} )^{*} $ 
follows immediately from the hypothesis of 
Theorem~\ref{annihilation-partial-answer}.  
\end{proof}

\begin{remark} 
The hypothesis in 
Theorem \ref{annihilation-partial-answer} 
that the measure 
$ \rho $ is radial is quite restrictive, of course. 
We do not know nor venture to conjecture what 
happens if that hypothesis is dropped. 
We also leave as open problems whether the number 
operator $ N $ or the creation operator 
$ T_{\theta} $ has an upper symbol. 
\end{remark}

\section{A Comparison}
\label{comparison-section}

We will now make a comparison with a finite dimensional 
algebra called the {\em paragrassmann algebra}, 
whose coherent states and their corresponding 
quantization were introduced and 
studied in \cite{baz}. 
The reproducing kernel object and the Toeplitz 
operators for this non-commutative space were 
investigated in \cite{Part-I} and \cite{Part-II}.
Coherent states were not introduced in \cite{sbs1}, 
but in that setting they would have been defined 
analogously to the coherent states introduced 
in \cite{baz} as a finite sum, 
namely as a formal infinite sum in Dirac notation
$$
  | \theta \rangle := \sum_{n=0}^{\infty} 
  w_{n}^{-1/2} \phi_{n} \otimes e_{n}, 
$$
where $\{ e_{n} ~|~ n \in \mathbb{N} \} $ is 
an orthonormal basis for an auxiliary Hilbert space.
This corresponds to \eqref{coherent-state} 
if we put $ \lambda = 1 $ there and take that also 
to be a formal infinite sum. 
But here we have what appears to be a better approach, 
since we do not fuss with making sense of 
formal infinite sums
but rather use convergent infinite series. 
This also has the advantage of providing a naturally 
defined phase space arising from the quantum theory.  
However, in \cite{baz} there is no phase space 
presented and that is in accord with our approach 
here as we now discuss. 

We start with a brief review of material in 
\cite{Part-II} where more details can be found. 
First, fix $ q \in \mathbb{C} \setminus \{ 0 \}$ 
and an integer $ l \ge 2 $. 
Then the paragrassmann algebra $ PG_{l,q} $
is defined as the quotient of the Manin plane 
$ \mathbb{C} Q_{q} (\theta, \overline{\theta})$
by putting $ \theta^{l} =0 $ and 
$ \overline{\theta}\,^{l} = 0 $. 
Then using a finite sequence of weights $ w_{n} > 0 $ 
for $ 0 \le n \le l - 1$
one introduces a sesqui-linear form by 
$
   \langle \theta^{i} \overline{\theta}\,\!^{j}, 
   \theta^{k} \overline{\theta}\,\!^{l} \rangle := 
   w_{i+l} \, \delta_{i-j,k-l} \chi_{l} (i+l),
$
where $ \chi_{l} (n) = 1 $ for $ 0 \le n \le l - 1$ 
and $ \chi_{l} (n) = 0 $ for $ n \ge l $. 
Here we have $ i,j,k,l \in \{ 1, 2, \dots, l-1 \} $. 
(Technically, we also have to put 
$ w_{n} =0 $ for $ n \ge l $.) 
Then this sesqui-linear form, when restricted to
the finite-dimensional 
sub-algebra~$ \mathcal{F} $ of $ PG_{l,q} $ 
generated by $ \theta $, satisfies
$$
\langle \theta^{i} , 
\theta^{k} \rangle := 
w_{i} \, \delta_{i,k} \chi_{l} (i) = 
w_{i} \, \delta_{i,k}
$$
and so is a positive definite inner product, 
thereby making the sub-algebra 
$ \mathcal{F} $ into a Hilbert space. 
There is an associated projection operator 
$ P : PG_{l,q} \to PG_{l,q} $, 
given in Dirac notation as 
$ P = \sum_{k=0}^{l-1} w_{k}^{-1}
 | \theta^{k} \rangle \langle  \theta^{k}  | $, 
whose range is $ \mathcal{F} $. 
Then a Toeplitz operator $ T_{g} $ with symbol 
$ g \in  PG_{l,q} $ is defined as expected:
$ T_{g} \, \phi := P(\phi g) $ for all 
$ \phi \in \mathcal{F} $. 
And it turns out that 
$ T_{g} : \mathcal{F} \to \mathcal{F} $ 
is a linear map. 
All of this is reminiscent of the Manin 
quantum plane except for details having to 
do with the fact that the variables $ \theta $ 
and $ \overline{\theta} $ are nilpotents. 
And a similar argument (see \cite{Part-II})
also shows that  
the annihilation operator $ T_{\overline{\theta}} $ 
satisfies 
\begin{equation}
\label{nilpotent}
T_{\overline{\theta}} \, \theta^{j}= 
\left( \dfrac{w_{j}}{w_{j-1}} \right)^{1/2}
\theta^{j-1}
\quad \mathrm{for~} 1 \le j \le l-1 
\quad \mathrm{and} \quad 
T_{\overline{\theta}} \, 1 = 0.
\end{equation} 
But we use $ \theta^{l} = 0 $ and
$ \theta^{l - 1} \ne 0 $ to see that 
$ ( T_{\overline{\theta}} )^{l} =0 $
and 
$ ( T_{\overline{\theta}} )^{l-1} \ne 0 $. 
So, $ T_{\overline{\theta}} $ is a nilpotent 
operator of nilpotency $ l $ acting on the 
Hilbert space $ \mathcal{F} $ whose dimension 
is $ l $. 
In fact \eqref{nilpotent} shows that 
$ T_{\overline{\theta}} $
is equivalent to one $ l \times l $ 
Jordan block with zeros along the diagonal. 
And its spectrum  is 
$ \mathrm{Spec} (T_{\overline{\theta}}) = \{ 0 \}$, 
with $\theta^{0} = 1 $ being an 
eigenvector of multiplicity one. 
Also, there are no other eigenvectors. 
So the phase space is the one-point set $\{ 0 \} $, 
a truly trivial situation.
And $ \phi_{\lambda = 0} = 1 $ is the only coherent state 
for the paragrassmann algebra, 
again a trivial situation. 

So our approach gives rather curious results for 
the paragrassmann algebra $ PG_{l,q} $ 
in contrast to the 
more conventional example of the Manin plane. 
We wish to emphasize that the 
paragrassmann algebra $ PG_{l,q} $ for $ q \ne 1 $
is a non-commutative quantum theory 
whose phase space is trivial.  
Moreover, the Hilbert space $ \mathcal{F} $ is not 
spanned by the coherent states, but far from it. 
According to Definition~\ref{define-extreme-quantum} 
the paragrassmann algebra $ PG_{l,q} $ 
is an extreme quantum theory. 
On the other hand, $ PG_{l,q} $ 
(even for the commutative case $q=1$)
is a quantum space with only one point,
that is, with exactly one unital algebra morphism 
$\alpha : PG_{l,q} \to \mathbb{C}$, 
since the nilpotency conditions force 
$ \alpha (\theta) = \alpha (\overline{\theta}) = 0 $.

\section{Concluding Remarks}

Here we have introduced the coherent states 
associated with just 
one Toeplitz annihilation operator. 
Rather analogous 
results should hold in settings where there 
is a commuting family of Toeplitz annihilation 
operators whose coherent states are defined as their 
common eigenvectors and whose phase space 
consists of their ordered $ n $-tples of eigenvalues. 
A more interesting situation arises if one is dealing 
with a family of non-commuting Toeplitz annihilation 
operators and their coherent states. 
Other possibilities for further research on this
topic include studying the semi-classical limit
and the minimal uncertainty of the coherent states. 
Also, the role of the Manin plane can be played 
by other non-commutative planes, and we will consider  
that topic in a forthcoming paper.

\end{document}